%% file: main.tex
\documentclass[pra,superscriptaddress,notitlepage,twocolumn]{revtex4-1}
\usepackage{amsmath,amsfonts,amssymb}
\usepackage{mathtools}
\usepackage{multirow}
\usepackage{ulem}

\usepackage[utf8]{inputenc}

\usepackage[T1]{fontenc}

\usepackage{bm}
\usepackage{times}
\usepackage{algorithm,algorithmic}

\input{pretex}

\definecolor{beamer}{rgb}{0.2,0.2,0.7}
\definecolor{colorone}{rgb}{1,0.36,0.03}
\definecolor{colortwo}{rgb}{0.4,0.77,0.17}
\definecolor{colorthree}{rgb}{0.01,0.51,0.93}
\definecolor{colorfour}{rgb}{0.47,0.26,0.58}
\definecolor{colorfive}{rgb}{0.12,0.55,0.16}

\usepackage{tcolorbox}
\usepackage{relsize}
\usepackage{graphicx}
\usepackage{booktabs}
\usepackage{amsmath}
\usepackage{tikz}
\usetikzlibrary{quantikz}
\nc{\st}{\text{subject to} \ }
\nc{\supre}{\text{supremum} \ }
\nc{\sdp}{\text{sdp}}
\usepackage{array}

\allowdisplaybreaks

\nc{\ith}[1]{{#1}^\mathrm{th}}

\begin{document}
\title{Variational quantum Hamiltonian engineering}
 
\author{Benchi Zhao}
\email{benchizhao@gmail.com}
\affiliation{Graduate School of Engineering Science, Osaka University, 1-3 Machikaneyama, Toyonaka, Osaka 560-8531, Japan}

\author{Keisuke Fujii}
\email{fujii@qc.ee.es.osaka-u.ac.jp}
\affiliation{Graduate School of Engineering Science, Osaka University, 1-3 Machikaneyama, Toyonaka, Osaka 560-8531, Japan}
\affiliation{Center for Quantum Information and Quantum Biology, Osaka University, 1-2 Machikaneyama, Toyonaka 560-8531, Japan}
\affiliation{RIKEN Center for Quantum Computing (RQC), Hirosawa 2-1, Wako, Saitama 351-0198, Japan}

\begin{abstract}
The Hamiltonian of a quantum system is represented in terms of operators corresponding to the kinetic and potential energies of the system. The expectation value of a Hamiltonian and Hamiltonian simulation are two of the most fundamental tasks in quantum computation. The overheads for realizing the two tasks are determined by the Pauli norm of Hamiltonian, which sums over all the absolute values of Pauli coefficients. In this work, we propose a variational quantum algorithm (VQA) called variational quantum Hamiltonian engineering (VQHE) to minimize the Pauli norm of Hamiltonian, such that the overhead for executing expectation value estimation and Hamiltonian simulation can be reduced. First, we develop a theory to encode the Pauli norm optimization problem into the vector $l_1$-norm minimization problem. Then we devise an appropriate cost function and utilize the parameterized quantum circuits (PQC) to minimize the cost function. We also conduct numerical experiments to reduce the Pauli norm of the Ising Hamiltonian and molecules' Hamiltonian to show the efficiency of the proposed VQHE. 
\end{abstract}

\maketitle

\section{Introduction}

Quantum computing, which harnesses the principles of quantum mechanics to revolutionize computation, is a promising solution to tackle problems that are currently intractable for classical computers. In recent decades, quantum technologies have emerged with a growing number of powerful applications across diverse fields such as optimization ~\cite{farhi2014quantum,harrigan2021quantum}, chemistry~\cite{mcardle2020quantum, google2020hartree}, security~\cite{bennett2014quantum,ekert1991quantum}, and machine learning~\cite{biamonte2017quantum}. In quantum computing, the Hamiltonian of a quantum system, an operator that represents the total energy of a quantum system, plays a pivotal role in many critical tasks such as expectation value estimation of a Hamiltonian and Hamiltonian simulation.

Estimating the expectation value of a Hamiltonian is a fundamental task in quantum mechanics and quantum computing, especially when simulating quantum systems or solving quantum chemistry problems. The expectation value of a Hamiltonian $H$ with respect to a quantum state $\ket{\psi}$ is given by $\langle H \rangle=\bra{\psi}H\ket{\psi}$. This quantity represents the average energy of the quantum state $\ket{\psi}$ when measured on the basis of the eigenstates of the Hamiltonian $H$. The most straightforward method to estimate the expectation value is to decompose the Hamiltonian into Pauli terms, i.e., $H=\sum_{i=1}^L h_i P_i$, and estimate the expectation value of each term $\langle P_i \rangle$. Then calculate the expectation value by post-processing $\langle H \rangle=\sum_i h_i \langle P_i \rangle$. The total measurement time will be $N\sim L^2|h_{\rm max}^2|$~\cite{kandala2017hardware,mcclean2016theory}, where $L$ is the number of non-zero Pauli coefficients, and $h_{\rm max}$ is the largest Pauli coefficient. While the number of Pauli terms typically increases polynomially in the system size, it amounts to be large for complex quantum systems such as those huddled in quantum chemistry, which makes this method no longer efficient.
To tackle this problem, one can sample the Pauli term $P_i$ of the target Hamiltonian $H$ randomly with a certain probability $p_i$, which is proportional to its absolute value of the coefficient, i.e., $p_i=\frac{|h_i|}{\sum_i|h_i|}$~\cite{wecker2015progress,huggins2021efficient, rubin2018application}. In this case, the total measurement time is determined by the summation of the absolute value of Pauli coefficients, i.e., $N\sim \sum_{i=1}^L|h_i|$.

Hamiltonian simulation is another fundamental task in quantum computing, which simulates the dynamics of a quantum system. The standard methods for Hamiltonian simulation, such as Trotter-Suzuki decomposition~\cite{trotter1959product, suzuki1991general}, are practical for sparse Hamiltonian. The gate count of such a method, which quantifies the complexity of Hamiltonian simulation, depends on the number of Pauli terms in the Hamiltonian. If the system's Hamiltonian is not sparse, then the consumption of the standard method will be too large to be acceptable. To address this problem, Campbell proposed an approach called quantum stochastic drift protocol (qDrift)~\cite{campbell2019random}. This protocol weights the probability of gates by the corresponding interaction strength in the Hamiltonian, leading to a gate count independent of the number of terms in the Hamiltonian but the summation of the absolute value of Pauli coefficients, i.e., $\sum_{i=1}^L|h_i|$.

The complexity of both tasks, expectation value estimation and Hamiltonian simulation, are dependent on the summation of the absolute value of Pauli coefficients, which is called \textit{Pauli norm} (or \textit{stabilizer norm}) of Hamiltonian. For a given Hamiltonian $H=\sum_{i=1}^{4^n}h_i P_i$, the Pauli norm $\|H\|_p$ is defined by
\begin{equation}
    \|H\|_p = \sum_{i=1}^{4^n}|h_i|.
\end{equation}

It is natural to wonder if a method exists to reduce the Pauli norm of a given Hamiltonian $H$ such that the complexities of the two tasks can be mitigated further. Our answer to this question is positive. In this work, we proposed a variational quantum algorithm (VQA)~\cite{cerezo2021variational,mitarai2018quantum} called variational quantum Hamiltonian engineering (VQHE) to minimize the Pauli norm of Hamiltonian, and the complexity of the two tasks are reduced correspondingly.  VQA is a popular paradigm for near-term quantum applications, which uses a classical optimizer to train parameterized quantum circuits (PQCs)~\cite{benedetti2019parameterized} to achieve certain tasks. VQAs have been successfully applied to various tasks such as searching ground state~\cite{kandala2017hardware, peruzzo2014variational} and excited state~\cite{mcclean2017hybrid, nakanishi2019subspace}, quantum classification~\cite{schuld2020circuit,farhi2018classification}, quantum information analysis~\cite{chen2023near, chen2021variational, zhao2021practical}, and quantum data compression~\cite{romero2017quantum, cao2021noise}. In previous works~\cite{wang2021variational, zeng2021variational}, it was proposed to diagonalize the Hamiltonian by VQA. Such a strategy can reduce the complexity of the expectation value of Hamiltonian because all the Pauli terms of the diagonalized Hamiltonian commute with each other, and they can be measured simultaneously. However, the VQAs to diagonalize the Hamiltonian are too costly. In the algorithm in Ref.~\cite{wang2021variational}, we have to estimate each eigenvalue of the Hamiltonian in every iteration. The non-zero eigenvalue of the Hamiltonian might increase exponentially with the system size, making the algorithm impossible. In Ref.~\cite{zeng2021variational}, the first step of the algorithm to diagonalize the Hamiltonian is preparing the Hamiltonian thermal state, which is a QMA-hard~\cite{wang2021Gibbs} problem. Thus, the algorithm in Ref.~\cite{zeng2021variational} is not practical either.

In this work, we propose a pre-processing algorithm to engineer the Hamiltonian, called variational quantum Hamiltonian engineering (VQHE), to minimize the Pauli norm of the engineered Hamiltonian. 
With the engineered Hamiltonian, the processing of Hamiltonian such as expectation value estimation and Hamiltonian simulation becomes easier to implement. 
Specifically, we use a parameterized unitary $U(\boldsymbol{\theta})$ to reduce the Pauli norm of a given Hamiltonian $H$, i.e., $\min\|U(\boldsymbol{\theta})HU(\boldsymbol{\theta})^\dagger\|_p$. 
We first develop the theory to convert the Pauli norm optimization problem into the state $l_1$-norm minimization problem. Then design an appropriate cost function and minimize it variationally. The VQHE algorithm outputs the engineered Hamiltonian, whose Pauli norm is minimized.
We then display how to apply the engineered Hamiltonian to the tasks of the expectation value estimation and the Hamiltonian simulation. Especially in the task of expectation value estimation, we emphasize that the engineered Hamiltonian is compatible with grouping strategy~\cite{kandala2017hardware,crawford2021efficient, gokhale2019minimizing,verteletskyi2020measurement}, which is another method to reduce the measurement time. The numerical experiments are conducted by applying the VQHE algorithms to the Ising Hamiltonian and some molecules' Hamiltonian, which shows the effectiveness of our algorithm. 
This work proposes a variational quantum algorithm to reduce the measurement complexity of expectation value estimation and the gate count of Hamiltonian simulation, making a significant contribution to improving the efficiency of quantum computing.

\section{Theoretical framework}\label{sec:theoretical_framework}
In this section, we are going to encode the Hamiltonian into a quantum state and prove that applying a unitary channel on a Hamiltonian is equivalent to applying a unitary gate on the Hamiltonian state.

For a given $n$-qubit Hamiltonian $H$, we can always decompose it into Pauli basis as 
\begin{equation}\label{eq:Pauli_decomposition}
    H = \sum_{i=1}^{4^n} h_i P_i, 
\end{equation}
where $h_i$ is the real Pauli coefficient and the $P_i\in\{I,X,Y,Z\}^{\ox n}$ is $n$-qubit Pauli tensor product. The problem of our task is minimizing the Pauli norm of Hamiltonian $H$ by optimizing the unitary $U$, which can be written as 
\begin{equation}\label{eq:prime_problem}
    P_{\rm opt}(H) = \min_{U} \|UHU^\dagger\|_p.
\end{equation}


On quantum computers, it is not convenient to apply a unitary channel $\cU$ on the Hamiltonian $H$ directly. For simplicity, our first step is to encode the Hamiltonian into a quantum state by the following Definition~\ref{def:vectorization}. 

\begin{definition}{\rm (Hamiltonian vectorization)}\label{def:vectorization} For a given $n$-qubit Hamiltonian $H=\sum_{i=1}^{4^n} h_i P_i$, it can be encoded into a $2n$-qubit state
    \begin{equation}
        \ket{H} = \frac{1}{\lambda} (h_1, h_2, \cdots, h_{4^n})^T,
    \end{equation}
where $\ket{H}$ is the Hamiltonian state, and Hamiltonian vector refers to the unnormalized Hamiltonian state, $T$ stands for transpose, and $\lambda$ refers to the normalization factor, i.e., $\lambda = \sqrt{\sum_{i=1}^{L=4^n} h_i^2}$. 
\end{definition}
For example, if the given Hamiltonian is $H=I+2X+3Y-4Z$, then the corresponding Hamiltonian state is $\ket{H} = \frac{1}{\sqrt{30}}(1, 2, 3, -4)^T$. The $l_1$-norm of the Hamiltonian state is $\|\ket{H}\|_1=10/\sqrt{30}$. The state initialization process can be efficiently executed on quantum device. It was proved in Ref.\cite{zhang2022quantum} that arbitrary $n$-qubit, $d$-sparse quantum state can be deterministically prepared with a circuit depth $\Theta(\log(nd))$ and $\cO(nd\log(d))$. Fortunately, the number of Pauli terms of the most meaningful Hamiltonian is polynomial, i.e., $\Tilde{\cO}(n^4)$. Thus the corresponding Hamiltonian state is $\Tilde{\cO}(n^4)$-sparse, implying it can be prepared efficiently.

The engineered Hamiltonian state can be directly represented as $\ket{H'} = \ket{UHU^\dagger}$. Such a state can be equivalently expressed as an \textit{encoded unitary} $V$ applied on the Hamiltonian state $\ket{H}$. The formal statement is shown in the following theorem.

\begin{theorem}\label{theo:Unitary}
    For a given Hamiltonian $H$ and unitary $U$, the vectorized engineered Hamiltonian is $\ket{H'} = \ket{UHU^\dagger}$, which is equivalent to applying a unitary gate $V$ on the Hamiltonian state
    \begin{equation}
        \ket{UHU^\dagger} = V\ket{H},
    \end{equation}
    with $V = \begin{bmatrix}
    \bra{c_1}\\
    \vdots\\
    \bra{c_{4^n}}
    \end{bmatrix}$, where $\ket{c_i}=\ket{U^\dagger P_i U}$ is the engineered Hamiltonian state of Pauli $P_i$.
\end{theorem}

\begin{proof}
For a given Hamiltonian $H$, we can decompose it into a linear combination of Pauli terms, i.e., $H=\sum_i h_i P_i$. The engineered Hamiltonian by unitary $U$ is $H' = UHU^\dagger = \sum_i h_i' P_i$. The engineered coefficients are obtained by following 
\begin{align}
    h_i' &= \frac{1}{2^n} \tr[UHU^\dagger P_i]\\
    &= \frac{1}{2^n} \tr[\sum_j h_j P_jU^\dagger P_i U]\\
    &= \frac{1}{2^n}\tr[\sum_j h_j P_j\sum_m c_{im} P_m]\\
    &= \frac{1}{2^n} \sum_{j,m} h_j c_{im}\tr[P_j  P_m]
\end{align}
Since $\tr[P_j P_m]=2^n\delta_{jm}$, where $P_j$ and $P_m$ are Pauli matrices, $n$ is number of qubits. Then,
\begin{align}
    h_i' &= \sum_{j} h_j c_{ij}=\langle c_i |H\rangle,
\end{align}
where $\ket{c_i}=\ket{U^\dagger P_i U}$ is the engineered Hamiltonian state of Pauli operator $P_i$. Here we can construct an operator $V$
\begin{equation}
    V = \begin{bmatrix}
    \bra{c_1}\\
    \bra{c_2}\\
    \vdots\\
    \bra{c_{4^n}}
    \end{bmatrix}
\end{equation}
such that 
\begin{equation}
    \ket{H'} = V\ket{H}.
\end{equation}

Since unitary $U$ preserves the $l_2$-norm, so \textbf{(1)} $\langle c_i|c_i\rangle=1$. 
Because $\tr[UP_iU^\dagger UP_jU^\dagger] = \tr[P_iP_j] = 0$ for $i \neq j$, then $\tr[\sum_m c_{im} P_m \sum_t c_{jt} P_t] = \sum_{m,t} c_{im} c_{jt} \tr[P_m P_t]= 0$, since $\tr[P_m P_t]= 2^n \delta_{mt}$, then it becomes $\sum_m c_{im} c_{jm}=0$, which implies \textbf{(2)} $\langle c_i |c_j \rangle =0, \forall i\neq j$
The two properties imply that the matrix $V$ is unitary.
\end{proof}

This theorem informs us that engineering Hamiltonian is equivalent to applying a unitary $V$ on the Hamiltonian vector. Specifically, from definition~\ref{def:vectorization}, the Hamiltonian state requires to normalize, so it is straightforward to have the relation between Pauli norm of Hamiltonian $\|H\|_p$ and the $l_1$-norm of Hamiltonian state 
\begin{equation}\label{eq:equivalence}
    \|UHU^\dagger\|_p = \lambda \|V\ket{H}\|_1.
\end{equation}
Since unitary preserves $l_2$-norm, meaning $\lambda$ is a constant, the problem of minimizing Pauli norm $\|H\|_p$ is equivalent to minimizing $l_1$-norm of Hamiltonian state $\|\ket{H}\|_1$, which is
\begin{equation}\label{eq:converted_problem}
    P_{\rm opt}(H) = \lambda\min_{V} \|V\ket{H}\|_1.
\end{equation}
Thus, when the optimal unitary $V$ is found to minimize the $l_1$-norm of $V\ket{H}$, it is equivalent to the optimal $U$ such that $\|UHU^\dagger\|_p$ is minimized. The specific relation between $U$ and $V$ is shown in Sec.~\ref{sec:Parameterized quantum circuit}.

Note $V$ is unitary (linear operation), so we can derive the following properties.
\begin{corollary}
The properties of operation $V$
    \begin{align}
    (1) \quad & \ket{U_2U_1HU_1^\dagger U_2^\dagger} = V_2V_1\ket{H}\\
    (2) \quad & (\alpha V_1 + \beta V_2)\ket{H} = \alpha V_1 \ket{H} + \beta V_2 \ket{H}\\
    (3) \quad & \ket{U_1\otimes U_2 H U_1^\dagger \otimes U_2^\dagger} = V_1\otimes V_2 \ket{H}
\end{align}
\end{corollary}

\begin{proof}
    For property (1), we can denote the $\Tilde{H}=U_1 H U_1^\dagger$, then $\ket{U_2U_1 H U_1^\dagger U_2^\dagger} = \ket{U_2\Tilde{H} U_2^\dagger} = V_2\ket{\Tilde{H}}$. Apply the theorem again, $V_2\ket{\Tilde{H}} = V_2\ket{U_1 H U_1^\dagger} = V_2V_1\ket{H}$.
    For property (2), it is obtained straightforwardly from the linearity of unitary $V$.
    For property (3), suppose $H=\sum_{ij} h_{ij} P_i\ox P_j$, then we have 
    \begin{align}
        U_1\otimes U_2 H U_1^\dagger \otimes U_2^\dagger 
        = \sum_{ij} h_{ij} U_1 P_i U_1^\dagger\otimes U_2 P_j  U_2^\dagger
    \end{align}
    where $i,j \in \{1,2,3,4\}$ such that $P_i, P_j$ corresponds to Pauli operators $\{I,X,Y,Z\}$ respectively. From Theorem~\ref{theo:Unitary}, we have $UP_i U^\dagger = V \vec{a}_i$, where $\vec{a}_i$ refers to the basis vector, i.e., only position $i$ in $\vec{a}_i$ is 1 and other positions are all 0. In another word, $\ket{P_i} = \Vec{a}_i$. Correspondingly, $\ket{H} = \frac{1}{\sqrt{\sum_{ij} h_{ij}^2}}\sum_{ij} h_{ij} \vec{a}_i\ox \vec{a}_j$. Now, we have
    \begin{align}
        U_1\otimes U_2 H U_1^\dagger \otimes U_2^\dagger &= \sum_{ij} h_{ij} V_1 \vec{a}_i \otimes V_2  \vec{a}_j \\
        &= V_1\otimes V_2 \sum_{ij} h_{ij} \vec{a}_i \otimes \vec{a}_j.
    \end{align}
    If we take Hamiltonian vectorization on both sides, it becomes
    \begin{equation}
        \ket{U_1\otimes U_2 H U_1^\dagger \otimes U_2^\dagger} = V_1\otimes V_2 \ket{H}.
    \end{equation}
    Note the fact that unitary is $l_2$- norm preserving operation, so the normalization factors are canceled out in this equation.
\end{proof}

\section{Variational quantum algorithms}
\begin{figure*}[t]
    \centering
    \includegraphics[width=\textwidth]{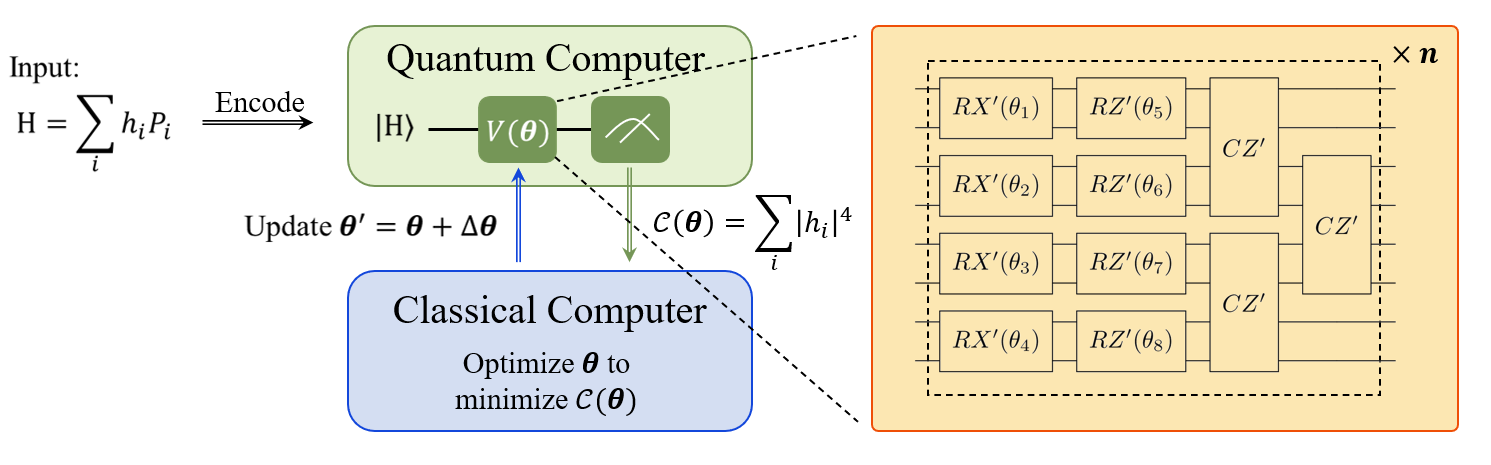}
    \caption{Diagram of the variational quantum algorithm for Hamiltonian engineering. Encode the given Hamiltonian $H$ into quantum state $\ket{H}$, then apply parameterized quantum circuit $V(\boldsymbol{\theta})$. Note that the rotation gates $RX'(\theta), RZ'(\theta)$, and control-Z gates $CZ'$ are transferred gates, the exact circuit representation is shown in Fig.~\ref{fig:gates}. Measuring the circuit and calculating the cost function, which is the $l_1$-norm of the quantum state, the classical computer optimizes the cost function by updating the parameters in the PQC. After iterations, the cost function converges to the optimized value.}
    \label{fig:main_diagram}
\end{figure*}

It is well-known that due to the non-convexity of unitary, it is very hard to obtain the optimal unitary to satisfy our requirement as shown in Eq.~\eqref{eq:converted_problem}. Here we propose to apply the variational quantum algorithm (VQA) to approach the optimal unitary. The workflow of our algorithm is shown in Fig.~\ref{fig:main_diagram}. At the beginning, we encode the Hamiltonian into a quantum state. Randomize the parameters in the PQC, and estimate the cost function by measurement. Then we use a classical computer to optimize the parameters in the PQC and update the parameters correspondingly. Repeat this process until the cost function is minimized. 

In this section, we are going to present the details of the cost function and the choice of PQC in the algorithm design. The specific algorithm is also provided as shown in Algorithm~\ref{algo:VQHE}.

\subsection{Cost function}
We have shown that the complexities of many tasks depend on the Pauli norm of the Hamiltonian, such as expectation value estimation~\cite{crawford2021efficient} and Hamiltonian simulation~\cite{campbell2019random}. We then convert the problem of minimizing the Pauli norm (Eq.~\eqref{eq:prime_problem}) to the problem of minimizing the $l_1$-norm of Hamiltonian vector (Eq.~\eqref{eq:converted_problem}). It is natural to consider the $l_1$-norm of the corresponding Hamiltonian state as the cost function, i.e., $\|\ket{H}\|_1$. However, estimating the $l_1$-norm of a quantum state is generally hard on quantum devices. Fortunately, there is a quantum algorithm that can efficiently estimate the sum of the fourth power of elements in the Hamilton vector $\ket{H}$, i.e.,
\begin{equation}
    Q = \sum_{i=0}^{d-1}|\langle i|H \rangle|^4 = \sum_{i=0}^{d-1}|h_i|^4
\end{equation}
where $d$ is the dimension of the system. The detail of the algorithm to estimate the quantity $Q$ is shown in the Appendix~\ref{app:estimate_Q}. This algorithm is also used in inverse participation ratio estimation~\cite{liu2024quantum}. 
Both $l_1$-norm and $Q$ describe the uncertainty of a quantum state. From this point of view, minimizing $l_1$-norm and maximizing $Q$ have the same optimizing direction, which reduces the uncertainty of a quantum state. Thus, we set $Q$ as the cost function, i.e., 
\begin{equation}\label{eq:cost_f}
    \cC = Q = \sum_{i=0}^{d-1}|h_i'|^4 
\end{equation}
Although minimizing $l_1$-norm and maximizing $Q$ is not exactly the same problem, the effectiveness of setting $Q$ as the cost function emerges in the numerical experiment in Sec.~\ref{sec:numerial_experiments}.

\subsection{Parameterized quantum circuit}\label{sec:Parameterized quantum circuit}
\begin{figure}[h!]
    \centering
    \includegraphics[width=0.5\textwidth]{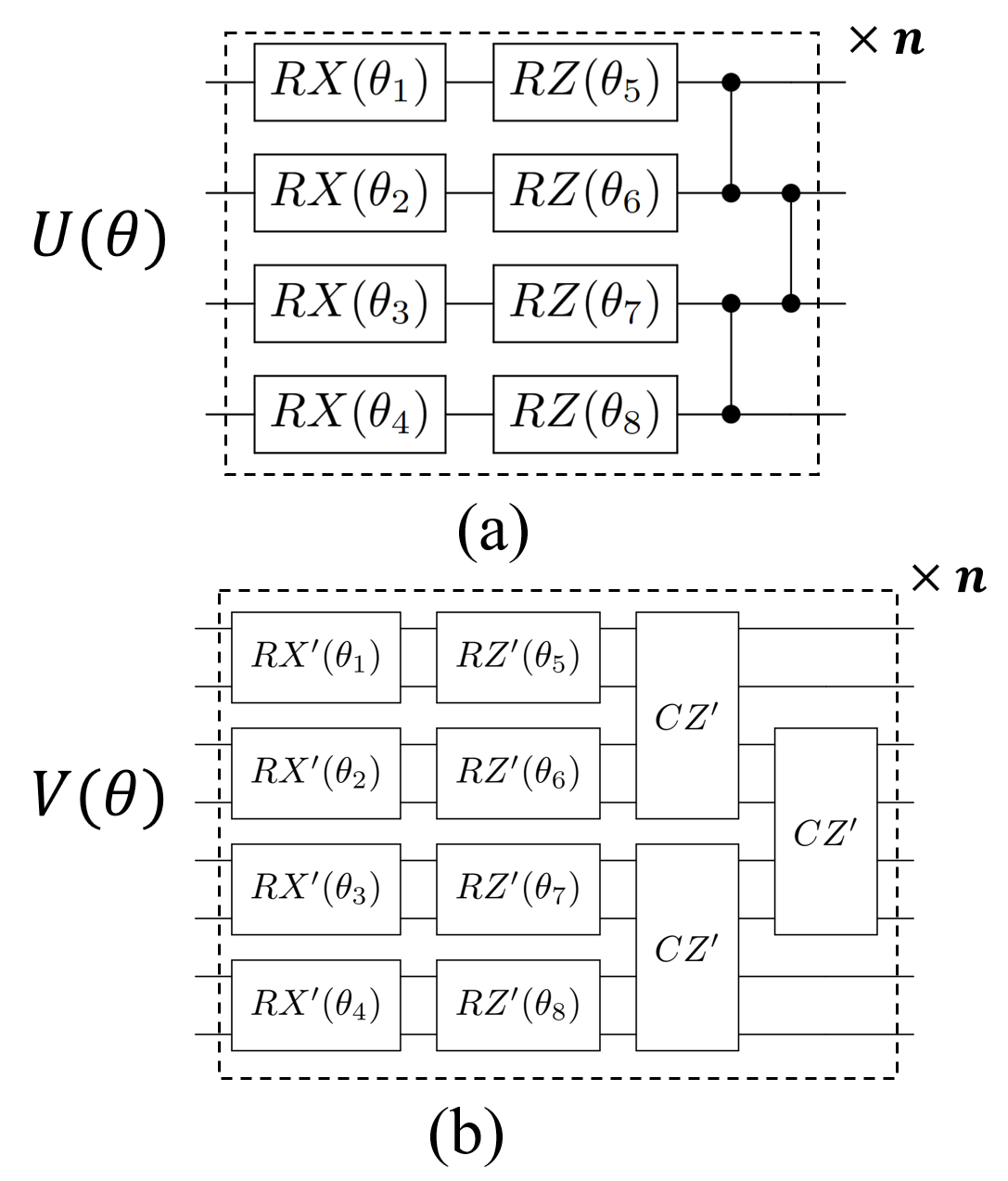}
    \caption{Structure of one entangled layer. (a) is the ansatz $U(\bm{\theta})$ for training Hamiltonian $H$. (b) is the 
    ansatz for training Hamiltonian state $\ket{H}$. The circuits for RX', RY', and CZ' are displayed in Appendix~\ref{app:Gate_correspondence}.}
    \label{fig:ansatz}
\end{figure}

From Theorem~\ref{theo:Unitary}, the encoded unitary $V$ is constructed from the unitary $U$, so $U$ and $V$ have a one-to-one correspondence relation. We can consider unitary $V$ as a function of unitary $U$, i.e., $V = f(U)$. The function $f$ maps an $n$-qubit unitary $U$ into a $2n$-qubit unitary $V$. Here we take the hardware-efficient ansatz for $U$ as shown in Fig.~\ref{fig:ansatz} (a), and we automatically obtain the corresponding encoded ansatz $V$ as shown in Fig.~\ref{fig:ansatz} (b). The ansatz for $U$ contains single qubit rotation gates rotation-X (RX), rotation-Z (RZ), and entangled gates control-Z (CZ). The corresponding encoded unitary components $V$ are denoted as RX', RZ', and CZ', whose detailed information is shown in Appendix~\ref{app:Gate_correspondence}. The encoded unitaries are all real, thus there is no need to worry about the imaginary element in the optimized Hamiltonian state. There is one thing we need to note that $U(\boldsymbol{\theta})$ and $V(\boldsymbol{\theta})$ share the same parameters $\boldsymbol{\theta}$. Once the parameters optimization in unitary $V(\boldsymbol{\theta})$ is done, we can obtain the unitary $U(\boldsymbol{\theta})$ directly by inserting the parameters $\boldsymbol{\theta}$ into $U$. The unitary $U$ is necessary when we apply the algorithm to expectation value estimation and Hamiltonian simulation, which will be discussed in Sec.~\ref{sec:application}.

\begin{algorithm}[H]
    \renewcommand{\algorithmicrequire}{\textbf{Input:}}
    \renewcommand{\algorithmicensure}{\textbf{Output:}}
    \begin{algorithmic}[1]
        \REQUIRE Classical description of $n$-qubit Hamiltonian $H=\sum_i h_i P_i$, parameterized quantum circuit (PQC) $V(\boldsymbol{\theta})$, number of iterations ITR,
        \ENSURE The optimized parameters $\boldsymbol{\theta}^*$, engineered Hamiltonian state $\ket{H'}$
        \STATE Calculate the normalization factor of the Pauli coefficients $\lambda=\sum_i |h_i|$
        \STATE Encode Hamiltonian $H$ into quantum state $\ket{H}$
        \STATE Randomly initialize parameters $\theta$
        \FOR {itr = 1, $\dots$, ITR}
            \STATE Apply $V(\boldsymbol{\theta})$ to the Hamiltonian state$\ket{H}$
            \STATE Apply the estimation algorithm (Appendix~\ref{app:estimate_Q}) to estimate the cost function in Eq.~\eqref{eq:cost_f}
            \STATE Maximize the cost function $\cC$ and update parameter $\boldsymbol{\theta}$
        \ENDFOR
        \STATE Output the optimized parameters $\boldsymbol{\theta^*}$, and the engineered Hamiltonian state $\ket{H'}$
    \end{algorithmic}
\caption{Variational quantum Hamiltonian engineering (VQHE)}
\label{algo:VQHE}
\end{algorithm}
\section{Application}\label{sec:application}
In this section, we are going to apply the proposed algorithm to reduce the measurement time in the task of expectation value estimation and the gate count in the task of Hamiltonian simulation, respectively. We especially emphasize that the proposed VQHE is compatible with grouping, which is another method to reduce the measurement complexity. The measurement complexity can be reduced for a further step. We also propose a partition trick to tackle the issue of scalability. 

To apply the engineered Hamiltonian $H'$ to the tasks of expectation value estimation and Hamiltonian simulation, we must first obtain the engineered Pauli coefficients $h_i'$. Quantum state tomography~\cite{nielsen2010quantum} is a candidate to get such information from the engineered Hamiltonian state $\ket{H'}$. However, the cost of state tomography increases exponentially with respect to the system size. Fortunately, the engineered Hamiltonian state $\ket{H'}$ is optimized by reducing its uncertainty, hence the complex amplitude is expected to be concentrated on a certain computational basis. In other words, we expect the engineered Hamiltonian state $\ket{H'}$ to be a sparse pure state. In Ref.~\cite{gulbahar2021k}, a more efficient algorithm than quantum state tomography was proposed to reconstruct the sparse pure state. With the engineered Hamiltonian coefficients $h_i'$, the $l_1$-norm of the engineered Hamiltonian norm can be calculated $\|H'\|_p=\sum_i |h_i|$ straightforwardly.

\subsection{Expectation value estimation}
In quantum computing, estimating the expectation value of a Hamiltonian $H$ to a quantum state, i.e., $\langle H\rangle =\bra{\psi}H \ket{\psi}$, is one of the most fundamental processes in many tasks. The total measurement time is determined by the Pauli norm, i.e., $N\sim \|H\|_p$. In order to reduce the measurement time, we can apply the proposed VQHE algorithm first as shown in Algorithm~\ref{algo:VQHE} as a pre-processing, which returns us the engineered Hamiltonian state $\ket{H'}$ and the optimized parameters $\boldsymbol{\theta^*}$. Then estimate the engineered Pauli coefficients $h_i'$ from the engineered Hamiltonian state $\ket{H'}$. The optimized Hamiltonian norm is obtained by summing up all the absolute values of the coefficients, i.e., $\|H'\|_p=\sum_i {p_i}$.

When we estimate the expectation value of Hamiltonian $H$, we have
\begin{equation}
    \bra{\psi}H\ket{\psi} = \bra{\psi}U^\dagger H' U\ket{\psi}.
\end{equation}
where $H' = UHU^\dagger$ is the engineered Hamiltonian. Here we set the unitary $U$ as $U(\boldsymbol{\theta^*})$, where $U(\boldsymbol{\theta^*})$ is the Ansatz shown in Fig.~\ref{fig:ansatz} (a), and the parameters $\boldsymbol{\theta^*}$ are optimized by Algorithm~\ref{algo:VQHE}.
In this case, the corresponding measurement time of estimating Hamiltonian $H$ reduces to $N\sim\|U(\boldsymbol{\theta}^*) H U^\dagger(\boldsymbol{\theta}^*)\|_p$.

The measurement time can be reduced further by applying a grouping algorithm to the engineered Hamiltonian. Grouping~\cite{kandala2017hardware, gokhale2019minimizing, verteletskyi2020measurement, crawford2021efficient} is an efficient classical algorithm to reduce the measurement time in the task of expectation value estimation. The main idea of grouping is dividing the Pauli terms into collections, such that all the Pauli elements in a collection commute with each other, thus they can be measured simultaneously. Many grouping methods have been proposed, such as qubit-wise commuting (QWC)~\cite{kandala2017hardware}, general commuting (GC)~\cite{gokhale2019minimizing, verteletskyi2020measurement}, and sorted insertion~\cite{crawford2021efficient}. Here we focus on the sorted insertion strategy, which is described as follows, and the details of QWC and GC can be found in the Appendix~\ref{app:Measurement_time_of_grouping}. 

For a given Hamiltonian as shown in Eq.~\eqref{eq:Pauli_decomposition}, the set $\{(h_i, P_i)\}_{i=1}^{4^n}$ is sorted by the absolute value of coefficients $h_j$ so that $|h_1|\ge|h_2|\ge\cdots\ge|h_{4^n}|$. Then, in the order $i = 1, \cdots 4^n$, it is checked whether $P_i$ commutes with all elements in an existing collection. If it does, it is added to that collection. If not, a new collection is created and $h_i P_i$ is inserted there. The collections are checked in order of their creation. After division, there are $N$ collections in total, and the $i$-th collection has $m_i$ Pauli terms. Here denote the grouped Pauli norm as $\|H\|_{gp} =\sum_{i=1}^N\sqrt{\sum_{j=1}^{m_i} |h_{ij}|^2}$, which is highly dependent on the group strategy. The total measurement time of the grouped collections is proportional to the square of the which is $N\sim \|H\|_{gp}$~\cite{crawford2021efficient}. The details and examples are shown in the Appendix~\ref{app:Grouped_Pauli_Norm}. If we apply the VQHE algorithm followed by the sorted grouping algorithm, the corresponding measurement time will be 
\begin{equation}
    N\sim \|U(\boldsymbol{\theta}^*) H U^\dagger(\boldsymbol{\theta}^*)\|_{gp}.
\end{equation}

\subsection{Hamiltonian simulation with qDrift}
Consider a Hamiltonian $H=\sum_{j=1}^Lh_j H_j$, where $H_j$ is hermitian and normalized, $h_j$ is the weight. For each $H_j$, the unitary $e^{-i\tau H_j}$ can be implemented on quantum computers for any $\tau$. In the Trotter formulae, one divides $U=e^{-itH}$ into segments so that $U = U^r_r$ with $U_r = e^{-itH/r}$ and uses $V_r=\prod_{i=1}^Le^{-ith_jH_j/r}$ to approach $U_r$ in the large $r$ limit. Repeat this process for $r$ times such that $V_r^r \rightarrow U$ in the large $r$ limit. It has been proven in Refs.~\cite{childs2019faster, childs2018toward} that the total gate count is $G\sim L^3(\Lambda t)^2/2\epsilon$, where $\Lambda:=\max_j h_j$ is the magnitude of the strongest term in the Hamiltonian, $\epsilon$ is the error tolerance. The gate count of this method depends on the terms of decomposition, which is not practical for electronic structure Hamiltonians due to its large number of terms.

The quantum stochastic drift protocol (qDrift)~\cite{campbell2019random} was proposed to solve this problem. Each unitary in the sequence is chosen independently from an identical distribution. Denote $\gamma=\sum_i h_j$. The strength $\tau_j$ of each unitary is fixed $\tau_j=t\gamma/G$, which is independent of $h_j$, so we implement gates of the form $e^{-i\tau H_j}$. Then we randomly choose the unitary $e^{-i\tau H_j}$ with probability $p_j = h_j/\gamma$. The full circuit is described by an ordered list of $j$ values $\bm{j}=\{j_1,j_2,\cdots,j_N\}$ that corresponds to unitary 
\begin{equation}
    V_{\bm{j}}=\prod_{k=1}^G e^{-i\tau H_{j_k}}.
\end{equation}
which is selected from the product distribution $P_{\bm{j}}
=\gamma^{-G}\prod_{k=1}^G h_{j_k}$. The gate count of Hamiltonian simulation by qDrift method is $G\sim \cO((\gamma t)^2/\epsilon)$, where $\gamma$ is the Pauli norm of Hamiltonian. 

Before simulating a given Hamiltonian $H$, we can first apply the VQHE algorithm to engineer the Hamiltonian such that the Pauli norm of the engineered Hamiltonian $H'=U(\boldsymbol{\theta}^*)HU^\dagger(\boldsymbol{\theta}^*)$ is reduced. Instead of simulating the original Hamiltonian, we can simulate the engineered Hamiltonian. Because of Eq.~\eqref{eq:Hamiltonian_sim_convert} 
\begin{equation}\label{eq:Hamiltonian_sim_convert}
    e^{-iHt} = U^\dagger(\boldsymbol{\theta}^*) e^{-iH't}U(\boldsymbol{\theta}^*),
\end{equation}
we need to sandwich the Hamiltonian evolution by unitaries. The corresponding gate count of simulating the engineered Hamiltonian is
\begin{equation}
    G\sim \cO((\gamma't)^2/\epsilon) + \Tilde{\cO}(U(\boldsymbol{\theta}^*)),
\end{equation}
where $\gamma'=\|U(\boldsymbol{\theta}^*) H U^\dagger(\boldsymbol{\theta}^*)\|_p$, and $\Tilde{\cO}(U(\boldsymbol{\theta}^*))$ is the gate count of the unitary, which can be ignored because it is a constant. The proof of Eq.~\eqref{eq:Hamiltonian_sim_convert} is shown in Appendix~\ref{app:proof_of_eq}.

\subsection{Scalability}
One obstacle to variational quantum algorithms is barren plateaus~\cite{mcclean2018barren,marrero2021entanglement, zhang2023statistical}, which refers to the phenomenon that the gradient decreases exponentially with the increase of the quantum system. To solve this problem, we propose a straightforward trick called \textit{partition trick}. Briefly speaking, we can divide the Hamiltonian into parts, where they have common factors. For each part, we extract the common factors and apply the VQHE to the rest Hamiltonian. We take the Hamiltonian $H=3XIXYYZ+2XIXZYI-XIXIIZ+IZXIYI+2IZZIYI$ as an example. We can first divide the Hamiltonian into two parts, and extract the common factor, thus we have $H_1= XIX\ox(3YYZ+2ZYI-IIZ)$ and $H_2 = (IZX+2IZZ)\ox IYI$. Then we apply VQHE to minimize $3YYZ+2ZYI-IIZ$ and $IZX+2IZZ$, respectively. The engineered Hamiltonian will be $H'=H_1' + H_2'$ with $H_1'=I\ox U_1 \cdot H_1 \cdot I\ox U_1^\dagger$ and similarly $H_2'= U_2\ox I \cdot H_2 \cdot U_2^\dagger\ox I$.

The partition trick is compatible with estimating expectation value. When we estimate the expectation value, we have
\begin{align}
    \bra{\psi} H \ket{\psi} = \sum_i \bra{\psi} H_i \ket{\psi} = 
    \sum_i \bra{\psi} U^\dagger_i H_i' U_i \ket{\psi}
\end{align}
where $H_i$ is the partite sub-Hamiltonian, and $H_i'=U_i H_i U_i^\dagger$ is the engineered sub-Hamiltonian. We should estimate the expectation values of each sub-Hamiltonian first, then sum them up. In such a case, the measurement time is $N\sim \sum_i\|H_i'\|_p$. In the task of Hamiltonian simulation, Eq.~\eqref{eq:Hamiltonian_sim_convert} doesn't hold anymore, for the unitaries $U_i$ are not the same for different partition Hamiltonian $H_i$. The partition trick is not compatible with Hamiltonian simulation.

The partition trick gets rid of barren plateaus by dividing a big optimization problem into several small ones, which helps us find the effective unitary to reduce the Pauli norm. However, such a trick will also reduce the expressibility of the ansatz, making us unlikely to obtain the optimal unitary to engineer Hamiltonian. This trick can be understood as a trade-off between effectiveness and optimality.

\subsection{Numerical experiments}\label{sec:numerial_experiments}
In this part, we are going to show the effectiveness of the proposed VQHE, which can reduce the Pauli norm of Hamiltonian. Specifically, we apply VQHE to the Ising Hamiltonian and the molecules' Hamiltonian. 

Ising model, a mathematical model of ferromagnetism in statistical mechanics, involves representing a material as a lattice of discrete points, where each point or "spin" can be in one of two possible states, typically denoted as "up" or "down." The spins interact with their nearest neighbors, and the interactions are characterized by a coupling constant. The Hamiltonian of the Ising model describes the energy of the system and is given by the sum of the interactions between neighboring spins. The basic form of the Hamiltonian for a one-dimensional Ising model is $ H = -J_{ij} \sum_{\{i,j\}} \sigma_i\sigma_j - g_k \sum_k \sigma_k$ where $\{i,j, k\}$ refers to the sit.

In the first numerical experiment, we apply the VQHE algorithm to the Ising model with different system sizes and compare the original Pauli norm and the engineered Pauli norm. For convenience, we set the coefficients in the Ising Hamiltonian to be the same, i.e., $J_{ij}=g_k=1$. In this experiment, we consider two Hamiltonians which are \textit{neighbor Ising Hamiltonian} $H_{\text{nei}}$ ($j=i+1$) and \textit{all-to-all Ising Hamiltonian} $H_{\text{all}}$ ($i$ and $j$ are independent), respectively. Specifically, we have
\begin{align}
    H_{\text{nei}} = -\sum_{\{i,i+1\}} Z_iZ_{i+1} + \sum_k X_k\\
    H_{\text{all}} = -\sum_{i,j} Z_iZ_j + \sum_k X_k,
\end{align}
where $X$ and $Z$ refer to Pauli-X and Pauli-Z operators, respectively. From the numerical results (Fig.~\ref{fig:Pauli_norm_of_ising}), the engineered Hamiltonian shows advantages over the original Hamiltonian with smaller Pauli norms in both cases of neighbor Ising Hamiltonian and all-to-all Ising Hamiltonian.

\begin{figure}[h]
    \centering
    \includegraphics[width=0.5\textwidth]{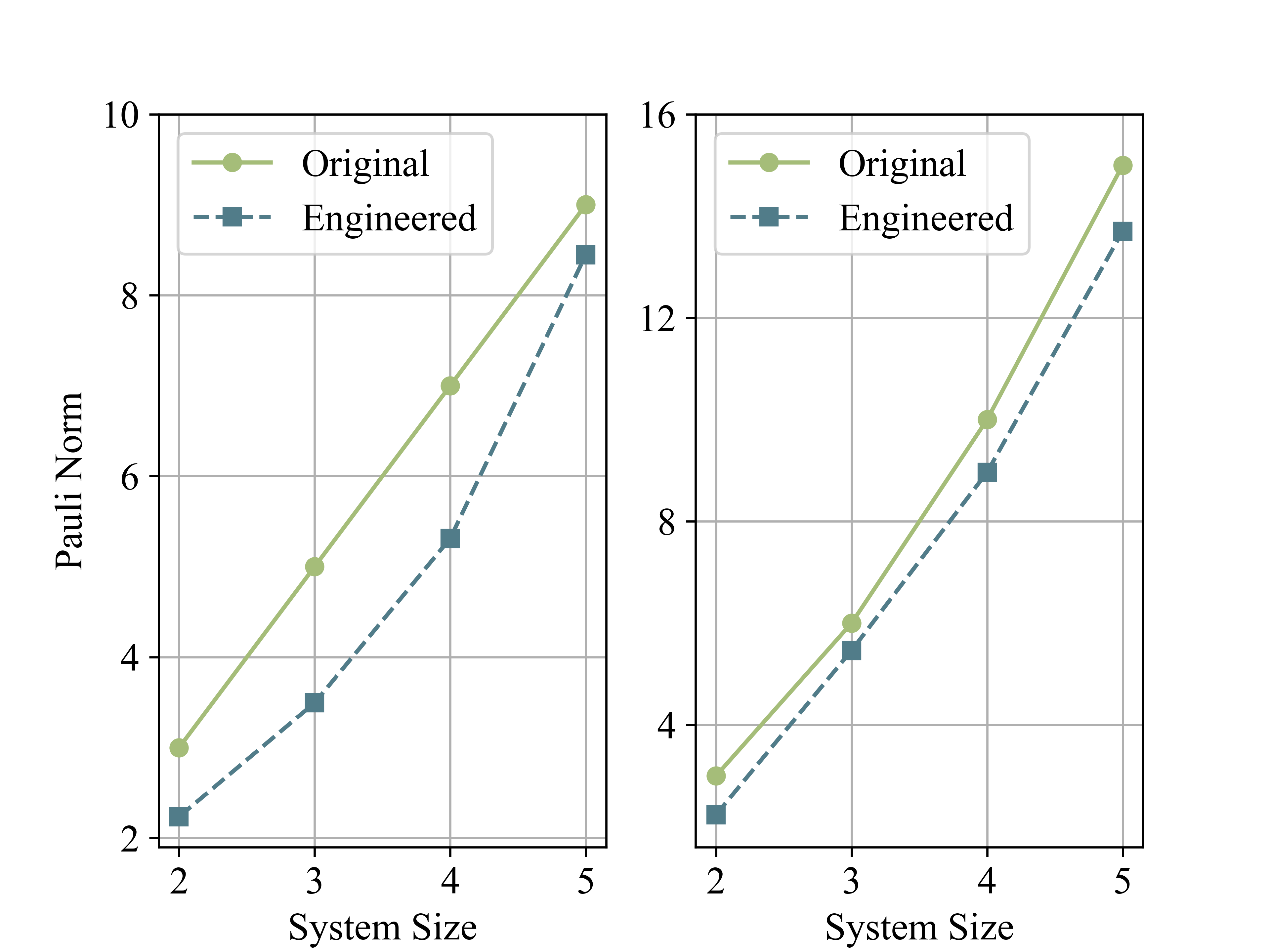}
    \caption{The Pauli norm comparison of different sizes of neighbor Ising Hamiltonian (left) and all-to-all Ising Hamiltonian (right). The solid line stands for the original Pauli norm, and the dashed line is the engineered Pauli norm.}
    \label{fig:Pauli_norm_of_ising}
\end{figure}

We also conduct numerical experiments on more complicated systems, including $H_2$, $HeH^+$, $H_3^+$, and $LiH$ molecules. In this experiment, we apply VQHE to the tapered Hamiltonian~\cite{bravyi2017tapering, setia2020reducing} of molecules for simplicity. The qubit tapering approach encodes the original quantum system into a smaller system, which preserves the ground state and ground energy. Thus, the tapered Hamiltonian reduces the requirement of the number of qubits, making the variational quantum eigensolver (VQE)~\cite{kandala2017hardware} more efficient. By applying VQHE to the tapered Hamiltonian of molecules, one observes the engineered Hamiltonian has a smaller Pauli norm compared with the original one. For a further step, we also take the sorted insertion grouping strategy into consideration. Specifically, we engineer the Hamiltonian with VQHE first, then apply the sorted insertion grouping method to the engineered Hamiltonian. The numerical results in Fig.~\ref{fig:Pauli_norm_of_molecules} show that the engineered grouped Pauli norm of Hamiltonian $\|H'\|_{gp}$ can be reduced further compared with that of engineered Hamiltonian $\|H'\|_{p}$ and that of grouped Hamiltonian $\|H\|_{gp}$. These results emphasize that the proposed VQHE is compatible with the grouping strategies.

\begin{figure}[h]
    \centering
    \includegraphics[width=0.5\textwidth]{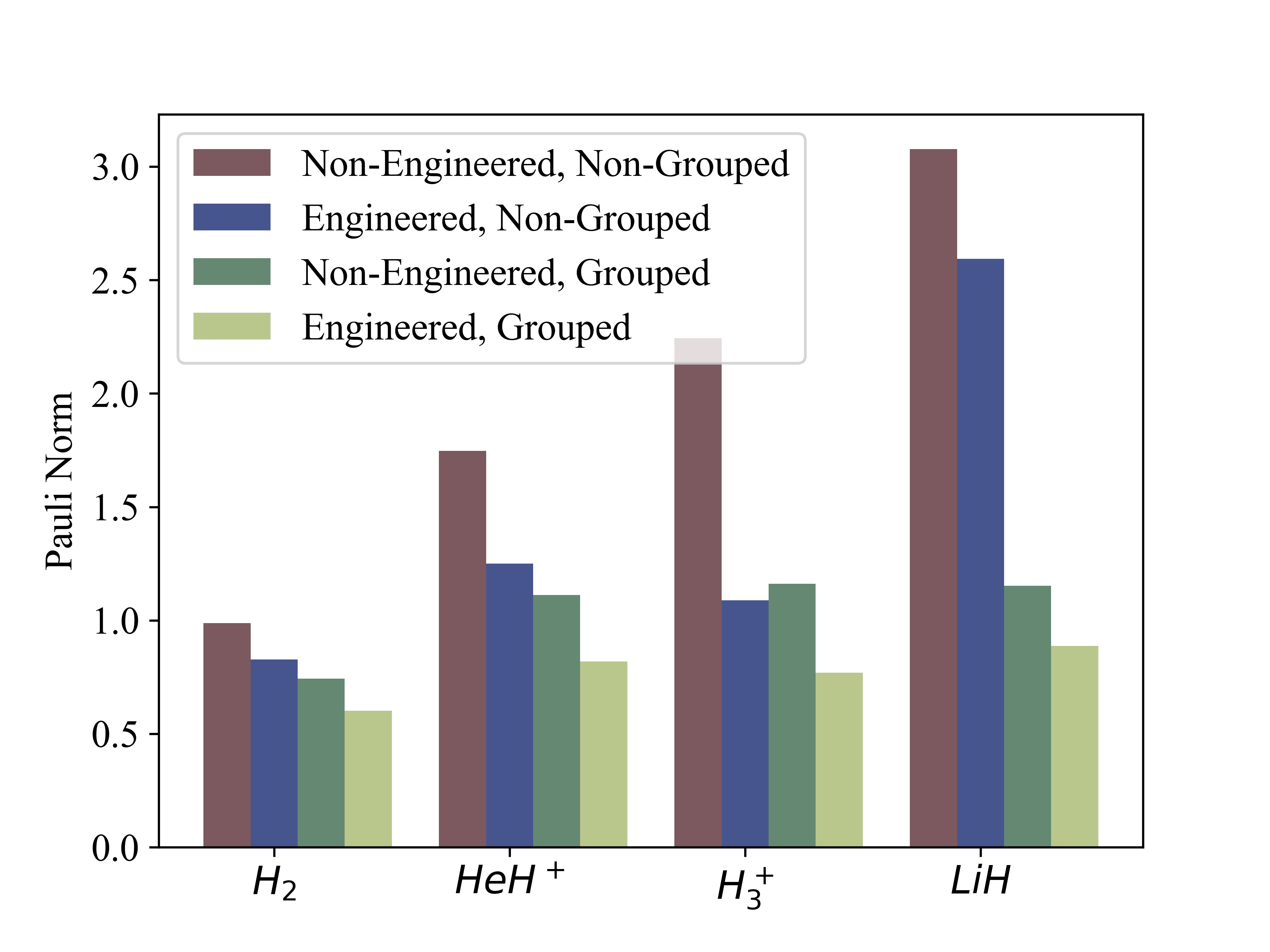}
    \caption{The Pauli norm of different Hamiltonian of molecules, including $H_2$~\cite{kandala2017hardware}, $HeH^+$, $H_3^+$~\cite{aguado2021three} and $LiH$~\cite{kandala2017hardware}. We compare the Pauli norm of the original Pauli norm, engineered Pauli norm, grouped Pauli norm~\cite{crawford2021efficient}, and engineered grouped Pauli norm}
    \label{fig:Pauli_norm_of_molecules}
\end{figure}

\section{Conclusion}
In this work, we propose the VQHE algorithm to reduce the Pauli norm of Hamiltonian such that the overheads for expectation value and Hamiltonian simulation can be reduced. We first develop the theory to convert the Pauli norm optimization problem into the vector $l_1$-norm minimization problem, and then design the cost function and parameterized quantum circuits to minimize Pauli norm variationally. We then display how to apply the proposed VQHE algorithm to expectation value estimation and Hamiltonian simulation. In the task of expectation value estimation, we also emphasize that the proposed algorithm is compatible with grouping, such that the measurement time can be reduced further. The numerical experiments are conducted by applying the VQHE algorithms to the Ising Hamiltonian and some molecules’ Hamiltonian, which shows the effectiveness of VQHE algorithm. 

For further research, it would also be fun to find the optimized Pauli norm of Hamiltonian we could obtain. It will also be interesting to apply the proposed algorithm VQHE to other applications in quantum computing. Besides, this framework may also be applied to tensor networks, empowering classical computers to solve quantum problems.


\textbf{Acknowledgments.--} B.Z would like to thank Hideaki Hakoshima,  Xuanqiang Zhao, and Xin Wang for their fruitful discussion. This work is supported by MEXT Quantum Leap Flagship Program (MEXT Q-LEAP) Grant No. JPMXS0118067394 and JPMXS0120319794, and JST COINEXT Grant No. JPMJPF2014.

\bibliography{references}

\clearpage

\vspace{2cm}
\onecolumngrid
\vspace{2cm}
\begin{center}
{\textbf{\large Supplementary Material for \\variational quantum Hamiltonian engineering}}
\end{center}

\appendix

\section{Quantum algorithm for estimating $Q$}\label{app:estimate_Q}
We denote the summation of the fourth power of elements in an $n$-qubit pure state $\ket{\psi}=(h_0, \dots, h_{d-1})^T$ as $Q$, which is 
\begin{equation}
    Q = \sum_{i=0}^{d-1}|\langle i|\psi\rangle|^4 =  \sum_{i=0}^{d-1} |h_i|^4
\end{equation}
where $\ket{i}$ refers to the computational basis. There exists an efficient quantum algorithm to estimate the quantity $Q$, which is shown in Fig.~\ref{fig:Q_algorithm}.

\begin{figure*}[htbp]
    \centering
    \includegraphics[width=0.7\textwidth]{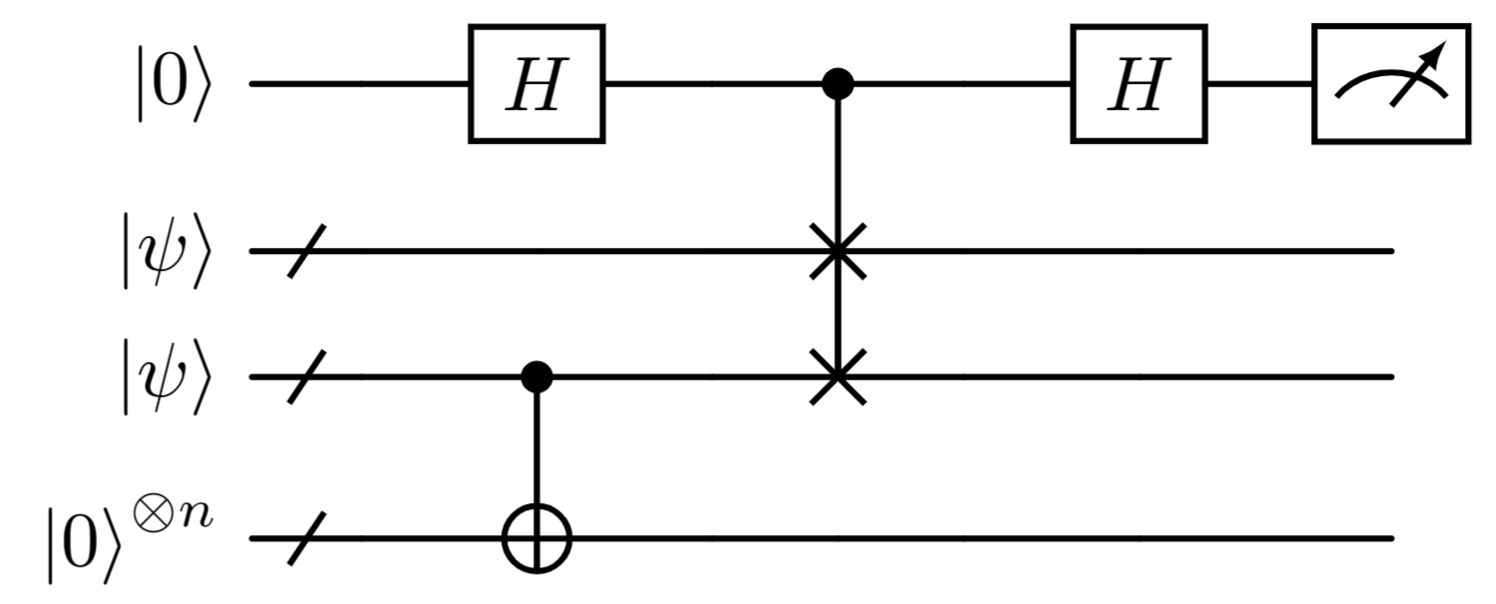}
    \caption{Quantum circuit for estimating $Q$. The quantum circuit consists of $n$ CNOT gates and $n$ controlled-SWAP gates. $Q$ is estimated through single-qubit measurements on the top qubit.}
    \label{fig:Q_algorithm}
\end{figure*}

The first step of the algorithm is to initialize the quantum state. We need four registers in total. The first register is called the measurement register, which is initialized as $\ket{0}$. The second and third registers are prepared as the $n$-qubit quantum state $\ket{\psi}$. The last register is prepared as $n$-qubit zero state $\ket{0}^{n}$. At the beginning, we have

\begin{align}
\ket{\Psi} = & \ket{0} \otimes |\psi\rangle \otimes |\psi\rangle \otimes |0\rangle^{ \otimes n}, \notag\\
= & \sum_{i} c_i |0\rangle \otimes | \psi \rangle \otimes |i\rangle \otimes |0\rangle^{ \otimes n}.
\end{align}

Apply CNOT gates on the third and fourth register, we then have 

\begin{align}\label{Eq4}
CNOT_{3,4} |\Psi_{0}\rangle
=\sum_{i} c_i |0\rangle \otimes | \psi \rangle \otimes |i\rangle \otimes |i\rangle.
\end{align}

At this step, we can discard the fourth register and the quantum state on the third register becomes a mixed state $\rho = \sum_{i} |c_i|^{2} |i\rangle\langle i|$. Take the SWAP test on the first three registers, the probability of getting the measurement as +1 is 
\begin{equation}\label{Eq5}
    P_{+1} = \frac{1}{2} + \frac{\tr[\rho\proj{\psi}]}{2} = \frac{1}{2} + \frac{Q}{2},
\end{equation}
where $Q=\sum_{i=0}^{d-1} |h_i|^4$ is the desired value. 

In Ref.~\cite{liu2024quantum}, A similar algorithm is also used in inverse participation ratio estimation.
\clearpage
\section{Gate correspondence}\label{app:Gate_correspondence}
\begin{figure*}[htbp]
    \centering
    \includegraphics[width=\textwidth]{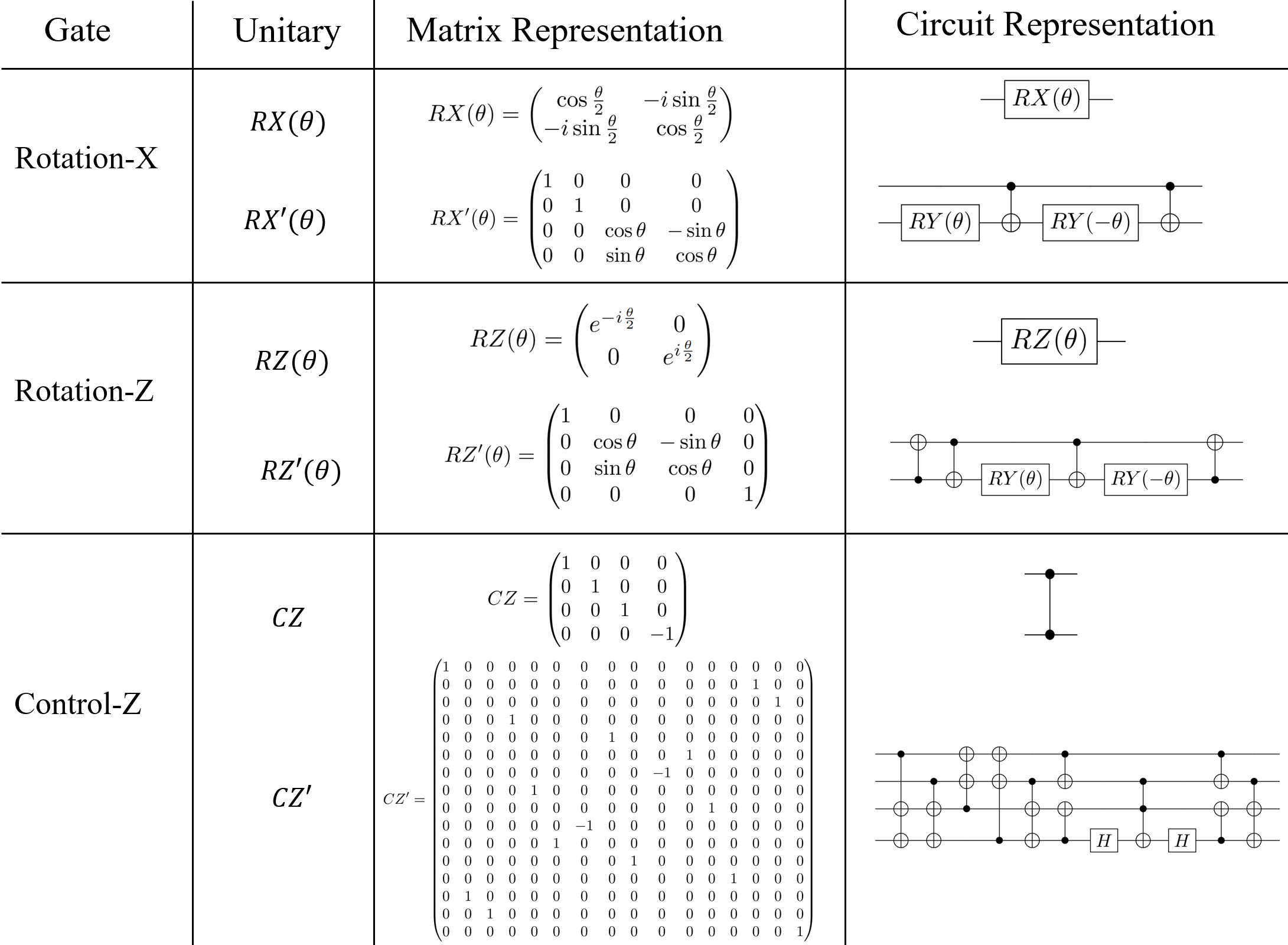}
    \caption{Unitary gates $U$ and the corresponding encoded unitary $V$}
    \label{fig:gates}
\end{figure*}

\section{Measurement time of grouping}~\label{app:Measurement_time_of_grouping}
In a normal distribution approximation, we require a variance of $\epsilon^2$ in the estimator $\langle H \rangle$. Then, the ground energy can be measured, 
\begin{equation}
    \langle H \rangle = \bra{\phi} H \ket{\phi} = \sum_i^L h_i \bra{\phi} H_i \ket{\phi} = \sum_i^L h_i \langle H_i \rangle.
\end{equation}
However, this straightforward method requires a huge amount of measurement time, making it infeasible to estimate the ground energy in large system. The standard error on the mean energy $\langle H \rangle$ after taking $S_i$ samples for each Pauli operator is 
\begin{equation}\label{eq:margin_of_error}
    \epsilon = \sqrt{\sum_i^L \frac{|h_i|^2 \text{Var}[H_i]}{S_i}}.
\end{equation}
If the measurement time of each term are the same $S_i= S$ \cite{kandala2017hardware, mcclean2016theory}, then the total measurement time is
\begin{equation}\label{eq:measure_times_independent}
    N_{measure} = ST = T\cdot\sum_i^L  \frac{|h_i|^2 \text{Var}[H_i]}{\epsilon^2}\le \frac{T^2|h_{max}^2|}{\epsilon^2}.
\end{equation}
where $h_{max} = \max_i|h_i|$ is the maximum coefficient.

If we choose $S_i \sim |h_i|$, and bounding the variances by Var$[H_i]\le 1$, then we get \cite{wecker2015progress, huggins2021efficient, rubin2018application}
\begin{equation}\label{eq:measure_times_dependent}
    N_{measure} =\frac{(\sum_i^L |h_i| \sqrt{\text{Var}[H_i]})^2}{\epsilon^2}\le  \frac{(\sum_i^L |h_i|)^2}{\epsilon^2},
\end{equation}
This inequality was proved in \cite{rubin2018application}. If we insist on getting an asymptotic bound then we assume Var$[H_i]=\cO(1)$, then the measurement time is optimal as in Eq.~\eqref{eq:measure_times_dependent}

After $S_i$ measurements, we get a set of values $\{x_j\}$ which is composed of 0 and 1, then the estimated expected value $\langle H \rangle$ and its variance $\text{Var}[H_i]$
\begin{align}
    \langle H_i \rangle &= \frac{1}{S_i} \sum_j^{S_i} x_j \\ 
    \text{Var}[H_i] &= \frac{1}{S_i-1}\sum_j^{S_i}(x_i - \langle H_i \rangle)^2
\end{align}

Clearly, when the system becomes complicated, it is hard to make measurement for the huge amount of resource requirements. From Eq.~\eqref{eq:measure_times_independent}, the fewer terms of the Hamiltonian, the less measurement time will be. Since commuting Pauli terms can be measured simultaneously, one straightforward idea is to partition the Pauli terms into groups, and in each group, all elements commute with each other. This kind of method is called \textit{Clique Cover}. If the number of such groups is minimum, it is called \textit{Minimum Clique Cover}.

\subsection{Qubit wise commuting (QWC)}
Qubit-wise commuting (QWC)\cite{kandala2017hardware} is one of the simplest methods to partition the Hamiltonian terms, which only group the terms that commute in qubit wise. For example $\{I_1 X_2, Z_1X_2\}$ is a QWC group because $I_1$ commutes with $X_1$ and $X_2$ commutes with $X_2$. While the case of $\{X_1X_2,Z_1Z_2\}$ is not a QWC group for $X_1$ and $Z_1$ do not commute. When doing the measurement, we only need to apply for local Clifford unitary $U$. This method can decrease 70\% terms.

\subsection{General commuting (GC)}
The general commuting method \cite{gokhale2019minimizing, verteletskyi2020measurement}  is a more powerful method compared with QWC. Note that QWC is simply the subset of GC corresponding to the case where the number of non-commuting indices is 0. When we make measurements, we need to apply the unitary $U$ transfer to the GC group into QWC group such that we can make measurements directly, i.e., 

\begin{align}
    \langle H \rangle &= \sum_i\bra{\phi} H_i \ket{\phi}\\
    &= \sum_i\bra{\phi} U_i^\dagger U_i H_i U_i^\dagger U_i \ket{\phi}\\
    &= \sum_i\bra{\phi} U_i^\dagger  A_i  U_i \ket{\phi}\\
    &=\sum_i\bra{\psi} A_i  \ket{\psi} = \sum_i \langle A_i \rangle,
\end{align}
where $A_i$ is the corresponding QWC group. The unitary $U$ can be decomposed into Clifford gates. This method can decrease two two-magnitude number of terms.

After partitioning into GC groups, we need to apply unitaries first, and then make simultaneous measurements. Since the unitaries are Clifford \cite{yen2020measuring}, the coefficient won't change after the unitary, therefore the

\section{The Grouped Pauli Norm}\label{app:Grouped_Pauli_Norm}
For a given operator $H$, we have
\begin{equation}
    H=\sum_{i=1}^N H_i = \sum_{i=1}^N\sum_{j=1}^{m_i} a_{ij}P_{ij},
\end{equation}
where $N$ is the number of collections of mutually commuting operators, $m_i$ the number of operators in collection $i$, $P_{ij}$ is the $j$ Pauli operator in the $i$th collection and $a_{ij}\in\mathbb{R}$ is its coefficient.

It has been proved~\cite{romero2018strategies, wecker2015progress,rubin2018application} that the optimal number of $m_i$ measurements of collection $i$ is 
\begin{equation}
    n_i = \frac{1}{\epsilon^2}\sqrt{\text{Var}[H_i]}\sum_{j=1}^N \sqrt{\text{Var}[H_j]}
\end{equation}

Therefore, the total number of measurements is
\begin{equation}
    M=\sum_i^N n_i = \Big(\frac{1}{\epsilon}\sum_{i=1}^N\sqrt{\text{Var}[H_i]}\Big)^2,
\end{equation}
where 
\begin{equation}
    {\rm Var}[H_i] = {\rm Cov}[H_i, H_i] = \langle H_i^2 \rangle - \langle H_i \rangle^2
\end{equation}

\begin{theorem}\cite{gokhale2019minimizing}\label{theorem: cov_to_zero}
Given $M_1, M_2$, two commuting but non-identical Pauli strings, $\mathbb{E}(\text{Cov}[M_1, M_2])=0$ where the expectation is taken over a uniform distribution over all possible state vectors.
\end{theorem}

For convenience, we assume we do not know the quantum state, and replace the variances and covariances by their expectation value over uniform spherical distribution. From Theorem~\ref{theorem: cov_to_zero}, it is straightforward to have 
\begin{equation}
    M=\sum_i^N n_i = \Bigg(\frac{1}{\epsilon}\sum_{i=1}^N\sqrt{\sum_{j=1}^{m_i} |h_{ij}|^2}\Bigg)^2
\end{equation}

To make a comparison with the ungrouped Pauli norm of Hamiltonian, we denote it as the grouped Pauli norm
\begin{equation}
    \|H\|_{gp} = \sum_{i=1}^N\sqrt{\sum_{j=1}^{m_i} |h_{ij}|^2}.
\end{equation}

Here we give a simple example by setting the Hamiltonian as $H=3XI -YY+ 2ZZ$. The first step is to sort it by the absolute value of coefficients, which will be $\{3XI, 2ZZ, -1YY\}$. At this stage, no collection is created, so we need to create one for the first element, such that $\{3XI\}$. For the second element $2ZZ$, we have to create another collection for $ZZ$ anti-commutes to $XI$. Now the collections will be $\{3XI\}, \{2ZZ\}$. For the third element $-1YY$, it does not commute with the element in the first collection, thus we need to check the commutativity of the second collection. Luckily, the third element commutes with all elements in the second collection, thus we can put it into this collection, which becomes $\{3XI\}, \{2ZZ, -1YY\}$. In this case, the corresponding grouped Pauli norm is $\|H\|_{gp}=5.24$, while the Pauli norm without grouping is $\|H\|_p = 6$.

\section{Proof of Eq.~\eqref{eq:Hamiltonian_sim_convert}}\label{app:proof_of_eq}
When we simulate the engineered Hamiltonian, we have to sandwich it with unitaries, to guarantee that the simulated Hamiltonian is the desired one,

\begin{align}
    U^\dagger e^{-iH't}U &=U^\dagger e^{-iUHU^\dagger t}U\\
    &=U^\dagger\sum_n \frac{(-it)^n}{n!}(UHU^\dagger)^n U\\
    &=U^\dagger\sum_n \frac{(-it)^n}{n!}(UHU^\dagger UHU^\dagger\cdots UHU^\dagger ) U\\
    &=U^\dagger\sum_n \frac{(-it)^n}{n!}(UH^nU^\dagger ) U\\
    &=U^\dagger U\sum_n \frac{(-it)^n}{n!}(H^n )U^\dagger U\\
    &=e^{-iHt}.
\end{align}

The second and last equations come from the Taylor expansion.




\end{document}

%% file: pretex.tex
\usepackage{mathtools}
\usepackage{amsmath}
\usepackage[shortlabels]{enumitem}

\usepackage{graphicx,epic,eepic,epsfig,amsmath,latexsym,amssymb,verbatim,color}
 
\usepackage{amsfonts}       
\usepackage{nicefrac}       

\usepackage{amsmath}
\usepackage{bbm}

\usepackage{float}
\usepackage{tikz}
\usetikzlibrary{chains}
\usetikzlibrary{fit}
\usetikzlibrary{arrows}

\usepackage{epsfig}
\usetikzlibrary{shapes.symbols,patterns} 
\usepackage{pgfplots}
\pgfplotsset{compat=1.18}

\usepackage[strict]{changepage}
\usepackage{hyperref}
\hypersetup{colorlinks=true,citecolor=blue,linkcolor=blue,filecolor=blue,urlcolor=blue,breaklinks=true}

\usepackage[marginal]{footmisc}
\usepackage{url}
\usepackage{theorem}

\newtheorem{definition}{Definition}
\newtheorem{proposition}{Proposition}

\newtheorem{theorem}[proposition]{Theorem}

\newtheorem{corollary}[proposition]{Corollary}

\def\squareforqed{\hbox{\rlap{$\sqcap$}$\sqcup$}}
\def\qed{\ifmmode\squareforqed\else{\unskip\nobreak\hfil
\penalty50\hskip1em\null\nobreak\hfil\squareforqed
\parfillskip=0pt\finalhyphendemerits=0\endgraf}\fi}
\def\endenv{\ifmmode\;\else{\unskip\nobreak\hfil
\penalty50\hskip1em\null\nobreak\hfil\;
\parfillskip=0pt\finalhyphendemerits=0\endgraf}\fi}
\newenvironment{proof}{\noindent \textbf{{Proof~} }}{\hfill $\blacksquare$}

\newcounter{remark}

\newcounter{example}

\mathchardef\ordinarycolon\mathcode`\:
\mathcode`\:=\string"8000
\def\vcentcolon{\mathrel{\mathop\ordinarycolon}}
\begingroup \catcode`\:=\active
  \lowercase{\endgroup
  \let :\vcentcolon
  }

\usepackage{cleveref}
\usepackage{graphicx}
\usepackage{xcolor}

\RequirePackage[framemethod=default]{mdframed}
\newmdenv[skipabove=7pt,
skipbelow=7pt,
backgroundcolor=darkblue!15,
innerleftmargin=5pt,
innerrightmargin=5pt,
innertopmargin=5pt,
leftmargin=0cm,
rightmargin=0cm,
innerbottommargin=5pt,
linewidth=1pt]{tBox}

\newmdenv[skipabove=7pt,
skipbelow=7pt,
backgroundcolor=blue2!25,
innerleftmargin=5pt,
innerrightmargin=5pt,
innertopmargin=5pt,
leftmargin=0cm,
rightmargin=0cm,
innerbottommargin=5pt,
linewidth=1pt]{dBox}
\newmdenv[skipabove=7pt,
skipbelow=7pt,
backgroundcolor=darkkblue!15,
innerleftmargin=5pt,
innerrightmargin=5pt,
innertopmargin=5pt,
leftmargin=0cm,
rightmargin=0cm,
innerbottommargin=5pt,
linewidth=1pt]{sBox}
\definecolor{darkblue}{RGB}{0,76,156}
\definecolor{darkkblue}{RGB}{0,0,153}
\definecolor{blue2}{RGB}{102,178,255}
\definecolor{darkred}{RGB}{195,0,0}

\newcommand{\nc}{\newcommand}
\nc{\rnc}{\renewcommand}
\nc{\lbar}[1]{\overline{#1}}
\nc{\bra}[1]{\langle#1|}
\nc{\ket}[1]{|#1\rangle}
\nc{\ketbra}[2]{|#1\rangle\!\langle#2|}
\nc{\braket}[2]{\langle#1|#2\rangle}

\nc{\proj}[1]{| #1\rangle\!\langle #1 |}
\nc{\avg}[1]{\langle#1\rangle}
\nc{\rank}{\operatorname{Rank}}
\nc{\smfrac}[2]{\mbox{$\frac{#1}{#2}$}}
\nc{\tr}{\operatorname{Tr}}
\nc{\ox}{\otimes}
\nc{\dg}{\dagger}
\nc{\dn}{\downarrow}
\nc{\cA}{{\cal A}}
\nc{\cB}{{\cal B}}
\nc{\cC}{{\cal C}}
\nc{\cD}{{\cal D}}
\nc{\cE}{{\cal E}}
\nc{\cF}{{\cal F}}
\nc{\cG}{{\cal G}}
\nc{\cH}{{\cal H}}
\nc{\cI}{{\cal I}}
\nc{\cJ}{{\cal J}}
\nc{\cK}{{\cal K}}
\nc{\cL}{{\cal L}}
\nc{\cM}{{\cal M}}
\nc{\cN}{{\cal N}}
\nc{\cO}{{\cal O}}
\nc{\cP}{{\cal P}}
\nc{\cQ}{{\cal Q}}
\nc{\cR}{{\cal R}}
\nc{\cS}{{\cal S}}
\nc{\cT}{{\cal T}}
\nc{\cU}{{\cal U}}
\nc{\cV}{{\cal V}}
\nc{\cX}{{\cal X}}
\nc{\cY}{{\cal Y}}
\nc{\cZ}{{\cal Z}}
\nc{\cW}{{\cal W}}
\nc{\csupp}{{\operatorname{csupp}}}
\nc{\qsupp}{{\operatorname{qsupp}}}
\nc{\var}{{\operatorname{var}}}
\nc{\rar}{\rightarrow}
\nc{\lrar}{\longrightarrow}
\nc{\polylog}{{\operatorname{polylog}}}
\nc{\wt}{{\operatorname{wt}}}
\nc{\av}[1]{{\left\langle {#1} \right\rangle}}
\nc{\supp}{{\operatorname{supp}}}

\nc{\argmin}{{\operatorname{argmin}}}

\def\x{\xi}

\nc{\RR}{{{\mathbb R}}}
\nc{\CC}{{{\mathbb C}}}
\nc{\FF}{{{\mathbb F}}}
\nc{\NN}{{{\mathbb N}}}
\nc{\ZZ}{{{\mathbb Z}}}
\nc{\PP}{{{\mathbb P}}}
\nc{\QQ}{{{\mathbb Q}}}
\nc{\UU}{{{\mathbb U}}}
\nc{\EE}{{{\mathbb E}}}
\nc{\id}{{\operatorname{id}}}

\nc{\CHSH}{{\operatorname{CHSH}}}

\nc{\be}{\begin{equation}}
\nc{\ee}{{\end{equation}}}
\nc{\bea}{\begin{eqnarray}}
\nc{\eea}{\end{eqnarray}}
\nc{\<}{\langle}
\rnc{\>}{\rangle}
\nc{\rU}{\mbox{U}}

\nc{\ob}[1]{#1}

\nc{\SEP}{{\text{\rm SEP}}}
\nc{\NS}{{\text{\rm NS}}}
\nc{\LOCC}{{\text{\rm LOCC}}}
\nc{\PPT}{{\text{\rm PPT}}}
\nc{\EXT}{{\text{\rm EXT}}}
\nc{\Sym}{{\operatorname{Sym}}}

\nc{\ERLO}{{E_{\text{r,LO}}}}
\nc{\ERLOCC}{{E_{\text{r,LOCC}}}}
\nc{\ERPPT}{{E_{\text{r,PPT}}}}
\nc{\ERLOCCinfty}{{E^{\infty}_{\text{r,LOCC}}}}
\nc{\Aram}{{\operatorname{\sf A}}}

\usepackage{tikz}
\usepackage{hyperref}
\hypersetup{colorlinks=true,citecolor=blue,linkcolor=blue,filecolor=blue,urlcolor=blue,breaklinks=true}

\makeatletter
\def\grd@save@target#1{%
  \def\grd@target{#1}}
\def\grd@save@start#1{%
  \def\grd@start{#1}}
\tikzset{
  grid with coordinates/.style={
    to path={%
      \pgfextra{%
        \edef\grd@@target{(\tikztotarget)}%
        \tikz@scan@one@point\grd@save@target\grd@@target\relax
        \edef\grd@@start{(\tikztostart)}%
        \tikz@scan@one@point\grd@save@start\grd@@start\relax
        \draw[minor help lines,magenta] (\tikztostart) grid (\tikztotarget);
        \draw[major help lines] (\tikztostart) grid (\tikztotarget);
        \grd@start
        \pgfmathsetmacro{\grd@xa}{\the\pgf@x/1cm}
        \pgfmathsetmacro{\grd@ya}{\the\pgf@y/1cm}
        \grd@target
        \pgfmathsetmacro{\grd@xb}{\the\pgf@x/1cm}
        \pgfmathsetmacro{\grd@yb}{\the\pgf@y/1cm}
        \pgfmathsetmacro{\grd@xc}{\grd@xa + \pgfkeysvalueof{/tikz/grid with coordinates/major step}}
        \pgfmathsetmacro{\grd@yc}{\grd@ya + \pgfkeysvalueof{/tikz/grid with coordinates/major step}}
        \foreach \x in {\grd@xa,\grd@xc,...,\grd@xb}
        \node[anchor=north] at (\x,\grd@ya) {\pgfmathprintnumber{\x}};
        \foreach \y in {\grd@ya,\grd@yc,...,\grd@yb}
        \node[anchor=east] at (\grd@xa,\y) {\pgfmathprintnumber{\y}};
      }
    }
  },
  minor help lines/.style={
    help lines,
    step=\pgfkeysvalueof{/tikz/grid with coordinates/minor step}
  },
  major help lines/.style={
    help lines,
    line width=\pgfkeysvalueof{/tikz/grid with coordinates/major line width},
    step=\pgfkeysvalueof{/tikz/grid with coordinates/major step}
  },
  grid with coordinates/.cd,
  minor step/.initial=.2,
  major step/.initial=1,
  major line width/.initial=2pt,
}
\makeatother

\usepackage{thmtools}
\usepackage{thm-restate}
\usepackage{etoolbox}
\makeatletter
\def\problem@s{}
\newcounter{problems@cnt}

\newcommand{\allproblems}{\problem@s}
\makeatother

%% file: main.bbl
\begin{thebibliography}{48}%
\makeatletter
\providecommand \@ifxundefined [1]{%
 \@ifx{#1\undefined}
}%
\providecommand \@ifnum [1]{%
 \ifnum #1\expandafter \@firstoftwo
 \else \expandafter \@secondoftwo
 \fi
}%
\providecommand \@ifx [1]{%
 \ifx #1\expandafter \@firstoftwo
 \else \expandafter \@secondoftwo
 \fi
}%
\providecommand \natexlab [1]{#1}%
\providecommand \enquote  [1]{``#1''}%
\providecommand \bibnamefont  [1]{#1}%
\providecommand \bibfnamefont [1]{#1}%
\providecommand \citenamefont [1]{#1}%
\providecommand \href@noop [0]{\@secondoftwo}%
\providecommand \href [0]{\begingroup \@sanitize@url \@href}%
\providecommand \@href[1]{\@@startlink{#1}\@@href}%
\providecommand \@@href[1]{\endgroup#1\@@endlink}%
\providecommand \@sanitize@url [0]{\catcode `\\12\catcode `\$12\catcode `\&12\catcode `\#12\catcode `\^12\catcode `\_12\catcode `\%12\relax}%
\providecommand \@@startlink[1]{}%
\providecommand \@@endlink[0]{}%
\providecommand \url  [0]{\begingroup\@sanitize@url \@url }%
\providecommand \@url [1]{\endgroup\@href {#1}{\urlprefix }}%
\providecommand \urlprefix  [0]{URL }%
\providecommand \Eprint [0]{\href }%
\providecommand \doibase [0]{http://dx.doi.org/}%
\providecommand \selectlanguage [0]{\@gobble}%
\providecommand \bibinfo  [0]{\@secondoftwo}%
\providecommand \bibfield  [0]{\@secondoftwo}%
\providecommand \translation [1]{[#1]}%
\providecommand \BibitemOpen [0]{}%
\providecommand \bibitemStop [0]{}%
\providecommand \bibitemNoStop [0]{.\EOS\space}%
\providecommand \EOS [0]{\spacefactor3000\relax}%
\providecommand \BibitemShut  [1]{\csname bibitem#1\endcsname}%
\let\auto@bib@innerbib\@empty
\bibitem [{\citenamefont {Farhi}\ \emph {et~al.}(2014)\citenamefont {Farhi}, \citenamefont {Goldstone},\ and\ \citenamefont {Gutmann}}]{farhi2014quantum}%
  \BibitemOpen
  \bibfield  {author} {\bibinfo {author} {\bibfnamefont {E.}~\bibnamefont {Farhi}}, \bibinfo {author} {\bibfnamefont {J.}~\bibnamefont {Goldstone}}, \ and\ \bibinfo {author} {\bibfnamefont {S.}~\bibnamefont {Gutmann}},\ }\href {https://arxiv.org/abs/1411.4028} {\bibfield  {journal} {\bibinfo  {journal} {arXiv preprint arXiv:1411.4028}\ } (\bibinfo {year} {2014})}\BibitemShut {NoStop}%
\bibitem [{\citenamefont {Harrigan}\ \emph {et~al.}(2021)\citenamefont {Harrigan}, \citenamefont {Sung}, \citenamefont {Neeley}, \citenamefont {Satzinger}, \citenamefont {Arute}, \citenamefont {Arya}, \citenamefont {Atalaya}, \citenamefont {Bardin}, \citenamefont {Barends}, \citenamefont {Boixo} \emph {et~al.}}]{harrigan2021quantum}%
  \BibitemOpen
  \bibfield  {author} {\bibinfo {author} {\bibfnamefont {M.~P.}\ \bibnamefont {Harrigan}}, \bibinfo {author} {\bibfnamefont {K.~J.}\ \bibnamefont {Sung}}, \bibinfo {author} {\bibfnamefont {M.}~\bibnamefont {Neeley}}, \bibinfo {author} {\bibfnamefont {K.~J.}\ \bibnamefont {Satzinger}}, \bibinfo {author} {\bibfnamefont {F.}~\bibnamefont {Arute}}, \bibinfo {author} {\bibfnamefont {K.}~\bibnamefont {Arya}}, \bibinfo {author} {\bibfnamefont {J.}~\bibnamefont {Atalaya}}, \bibinfo {author} {\bibfnamefont {J.~C.}\ \bibnamefont {Bardin}}, \bibinfo {author} {\bibfnamefont {R.}~\bibnamefont {Barends}}, \bibinfo {author} {\bibfnamefont {S.}~\bibnamefont {Boixo}},  \emph {et~al.},\ }\href {https://www.nature.com/articles/s41567-020-01105-y} {\bibfield  {journal} {\bibinfo  {journal} {Nature Physics}\ }\textbf {\bibinfo {volume} {17}},\ \bibinfo {pages} {332} (\bibinfo {year} {2021})}\BibitemShut {NoStop}%
\bibitem [{\citenamefont {McArdle}\ \emph {et~al.}(2020)\citenamefont {McArdle}, \citenamefont {Endo}, \citenamefont {Aspuru-Guzik}, \citenamefont {Benjamin},\ and\ \citenamefont {Yuan}}]{mcardle2020quantum}%
  \BibitemOpen
  \bibfield  {author} {\bibinfo {author} {\bibfnamefont {S.}~\bibnamefont {McArdle}}, \bibinfo {author} {\bibfnamefont {S.}~\bibnamefont {Endo}}, \bibinfo {author} {\bibfnamefont {A.}~\bibnamefont {Aspuru-Guzik}}, \bibinfo {author} {\bibfnamefont {S.~C.}\ \bibnamefont {Benjamin}}, \ and\ \bibinfo {author} {\bibfnamefont {X.}~\bibnamefont {Yuan}},\ }\href {https://journals.aps.org/rmp/abstract/10.1103/RevModPhys.92.015003} {\bibfield  {journal} {\bibinfo  {journal} {Reviews of Modern Physics}\ }\textbf {\bibinfo {volume} {92}},\ \bibinfo {pages} {015003} (\bibinfo {year} {2020})}\BibitemShut {NoStop}%
\bibitem [{\citenamefont {Quantum}\ \emph {et~al.}(2020)\citenamefont {Quantum}, \citenamefont {Collaborators*†}, \citenamefont {Arute}, \citenamefont {Arya}, \citenamefont {Babbush}, \citenamefont {Bacon}, \citenamefont {Bardin}, \citenamefont {Barends}, \citenamefont {Boixo}, \citenamefont {Broughton}, \citenamefont {Buckley} \emph {et~al.}}]{google2020hartree}%
  \BibitemOpen
  \bibfield  {author} {\bibinfo {author} {\bibfnamefont {G.~A.}\ \bibnamefont {Quantum}}, \bibinfo {author} {\bibnamefont {Collaborators*†}}, \bibinfo {author} {\bibfnamefont {F.}~\bibnamefont {Arute}}, \bibinfo {author} {\bibfnamefont {K.}~\bibnamefont {Arya}}, \bibinfo {author} {\bibfnamefont {R.}~\bibnamefont {Babbush}}, \bibinfo {author} {\bibfnamefont {D.}~\bibnamefont {Bacon}}, \bibinfo {author} {\bibfnamefont {J.~C.}\ \bibnamefont {Bardin}}, \bibinfo {author} {\bibfnamefont {R.}~\bibnamefont {Barends}}, \bibinfo {author} {\bibfnamefont {S.}~\bibnamefont {Boixo}}, \bibinfo {author} {\bibfnamefont {M.}~\bibnamefont {Broughton}}, \bibinfo {author} {\bibfnamefont {B.~B.}\ \bibnamefont {Buckley}},  \emph {et~al.},\ }\href {https://www.science.org/doi/full/10.1126/science.abb9811?casa_token=UzSCIxE34wwAAAAA:0JNzYJA8s5Z_rYbOxDxox6xaZ-pbXSc351kzlrMKJGcwQvlgK7_5jDlrBf7_8o00FHjr1xoB5aBfrRI} {\bibfield  {journal} {\bibinfo  {journal} {Science}\ }\textbf {\bibinfo {volume} {369}},\ \bibinfo {pages} {1084}
  (\bibinfo {year} {2020})}\BibitemShut {NoStop}%
\bibitem [{\citenamefont {Bennett}\ and\ \citenamefont {Brassard}(2014)}]{bennett2014quantum}%
  \BibitemOpen
  \bibfield  {author} {\bibinfo {author} {\bibfnamefont {C.~H.}\ \bibnamefont {Bennett}}\ and\ \bibinfo {author} {\bibfnamefont {G.}~\bibnamefont {Brassard}},\ }\href {https://www.sciencedirect.com/science/article/pii/S0304397514004241?via%3Dihub} {\bibfield  {journal} {\bibinfo  {journal} {Theoretical computer science}\ }\textbf {\bibinfo {volume} {560}},\ \bibinfo {pages} {7} (\bibinfo {year} {2014})}\BibitemShut {NoStop}%
\bibitem [{\citenamefont {Ekert}(1991)}]{ekert1991quantum}%
  \BibitemOpen
  \bibfield  {author} {\bibinfo {author} {\bibfnamefont {A.~K.}\ \bibnamefont {Ekert}},\ }\href {https://journals.aps.org/prl/abstract/10.1103/PhysRevLett.67.661} {\bibfield  {journal} {\bibinfo  {journal} {Physical review letters}\ }\textbf {\bibinfo {volume} {67}},\ \bibinfo {pages} {661} (\bibinfo {year} {1991})}\BibitemShut {NoStop}%
\bibitem [{\citenamefont {Biamonte}\ \emph {et~al.}(2017)\citenamefont {Biamonte}, \citenamefont {Wittek}, \citenamefont {Pancotti}, \citenamefont {Rebentrost}, \citenamefont {Wiebe},\ and\ \citenamefont {Lloyd}}]{biamonte2017quantum}%
  \BibitemOpen
  \bibfield  {author} {\bibinfo {author} {\bibfnamefont {J.}~\bibnamefont {Biamonte}}, \bibinfo {author} {\bibfnamefont {P.}~\bibnamefont {Wittek}}, \bibinfo {author} {\bibfnamefont {N.}~\bibnamefont {Pancotti}}, \bibinfo {author} {\bibfnamefont {P.}~\bibnamefont {Rebentrost}}, \bibinfo {author} {\bibfnamefont {N.}~\bibnamefont {Wiebe}}, \ and\ \bibinfo {author} {\bibfnamefont {S.}~\bibnamefont {Lloyd}},\ }\href {https://www.nature.com/articles/nature23474} {\bibfield  {journal} {\bibinfo  {journal} {Nature}\ }\textbf {\bibinfo {volume} {549}},\ \bibinfo {pages} {195} (\bibinfo {year} {2017})}\BibitemShut {NoStop}%
\bibitem [{\citenamefont {Kandala}\ \emph {et~al.}(2017)\citenamefont {Kandala}, \citenamefont {Mezzacapo}, \citenamefont {Temme}, \citenamefont {Takita}, \citenamefont {Brink}, \citenamefont {Chow},\ and\ \citenamefont {Gambetta}}]{kandala2017hardware}%
  \BibitemOpen
  \bibfield  {author} {\bibinfo {author} {\bibfnamefont {A.}~\bibnamefont {Kandala}}, \bibinfo {author} {\bibfnamefont {A.}~\bibnamefont {Mezzacapo}}, \bibinfo {author} {\bibfnamefont {K.}~\bibnamefont {Temme}}, \bibinfo {author} {\bibfnamefont {M.}~\bibnamefont {Takita}}, \bibinfo {author} {\bibfnamefont {M.}~\bibnamefont {Brink}}, \bibinfo {author} {\bibfnamefont {J.~M.}\ \bibnamefont {Chow}}, \ and\ \bibinfo {author} {\bibfnamefont {J.~M.}\ \bibnamefont {Gambetta}},\ }\href {https://www.nature.com/articles/nature23879.} {\bibfield  {journal} {\bibinfo  {journal} {Nature}\ }\textbf {\bibinfo {volume} {549}},\ \bibinfo {pages} {242} (\bibinfo {year} {2017})}\BibitemShut {NoStop}%
\bibitem [{\citenamefont {McClean}\ \emph {et~al.}(2016)\citenamefont {McClean}, \citenamefont {Romero}, \citenamefont {Babbush},\ and\ \citenamefont {Aspuru-Guzik}}]{mcclean2016theory}%
  \BibitemOpen
  \bibfield  {author} {\bibinfo {author} {\bibfnamefont {J.~R.}\ \bibnamefont {McClean}}, \bibinfo {author} {\bibfnamefont {J.}~\bibnamefont {Romero}}, \bibinfo {author} {\bibfnamefont {R.}~\bibnamefont {Babbush}}, \ and\ \bibinfo {author} {\bibfnamefont {A.}~\bibnamefont {Aspuru-Guzik}},\ }\href {https://iopscience.iop.org/article/10.1088/1367-2630/18/2/023023} {\bibfield  {journal} {\bibinfo  {journal} {New Journal of Physics}\ }\textbf {\bibinfo {volume} {18}},\ \bibinfo {pages} {023023} (\bibinfo {year} {2016})}\BibitemShut {NoStop}%
\bibitem [{\citenamefont {Wecker}\ \emph {et~al.}(2015)\citenamefont {Wecker}, \citenamefont {Hastings},\ and\ \citenamefont {Troyer}}]{wecker2015progress}%
  \BibitemOpen
  \bibfield  {author} {\bibinfo {author} {\bibfnamefont {D.}~\bibnamefont {Wecker}}, \bibinfo {author} {\bibfnamefont {M.~B.}\ \bibnamefont {Hastings}}, \ and\ \bibinfo {author} {\bibfnamefont {M.}~\bibnamefont {Troyer}},\ }\href {https://journals.aps.org/pra/abstract/10.1103/PhysRevA.92.042303} {\bibfield  {journal} {\bibinfo  {journal} {Physical Review A}\ }\textbf {\bibinfo {volume} {92}},\ \bibinfo {pages} {042303} (\bibinfo {year} {2015})}\BibitemShut {NoStop}%
\bibitem [{\citenamefont {Huggins}\ \emph {et~al.}(2021)\citenamefont {Huggins}, \citenamefont {McClean}, \citenamefont {Rubin}, \citenamefont {Jiang}, \citenamefont {Wiebe}, \citenamefont {Whaley},\ and\ \citenamefont {Babbush}}]{huggins2021efficient}%
  \BibitemOpen
  \bibfield  {author} {\bibinfo {author} {\bibfnamefont {W.~J.}\ \bibnamefont {Huggins}}, \bibinfo {author} {\bibfnamefont {J.~R.}\ \bibnamefont {McClean}}, \bibinfo {author} {\bibfnamefont {N.~C.}\ \bibnamefont {Rubin}}, \bibinfo {author} {\bibfnamefont {Z.}~\bibnamefont {Jiang}}, \bibinfo {author} {\bibfnamefont {N.}~\bibnamefont {Wiebe}}, \bibinfo {author} {\bibfnamefont {K.~B.}\ \bibnamefont {Whaley}}, \ and\ \bibinfo {author} {\bibfnamefont {R.}~\bibnamefont {Babbush}},\ }\href {https://www.nature.com/articles/s41534-021-00362-w} {\bibfield  {journal} {\bibinfo  {journal} {npj Quantum Information}\ }\textbf {\bibinfo {volume} {7}},\ \bibinfo {pages} {1} (\bibinfo {year} {2021})}\BibitemShut {NoStop}%
\bibitem [{\citenamefont {Rubin}\ \emph {et~al.}(2018)\citenamefont {Rubin}, \citenamefont {Babbush},\ and\ \citenamefont {McClean}}]{rubin2018application}%
  \BibitemOpen
  \bibfield  {author} {\bibinfo {author} {\bibfnamefont {N.~C.}\ \bibnamefont {Rubin}}, \bibinfo {author} {\bibfnamefont {R.}~\bibnamefont {Babbush}}, \ and\ \bibinfo {author} {\bibfnamefont {J.}~\bibnamefont {McClean}},\ }\href {https://iopscience.iop.org/article/10.1088/1367-2630/aab919} {\bibfield  {journal} {\bibinfo  {journal} {New Journal of Physics}\ }\textbf {\bibinfo {volume} {20}},\ \bibinfo {pages} {053020} (\bibinfo {year} {2018})}\BibitemShut {NoStop}%
\bibitem [{\citenamefont {Trotter}(1959)}]{trotter1959product}%
  \BibitemOpen
  \bibfield  {author} {\bibinfo {author} {\bibfnamefont {H.~F.}\ \bibnamefont {Trotter}},\ }\href {http://dx.doi.org/10.1090/S0002-9939-1959-0108732-6} {\bibfield  {journal} {\bibinfo  {journal} {Proceedings of the American Mathematical Society}\ }\textbf {\bibinfo {volume} {10}},\ \bibinfo {pages} {545} (\bibinfo {year} {1959})}\BibitemShut {NoStop}%
\bibitem [{\citenamefont {Suzuki}(1991)}]{suzuki1991general}%
  \BibitemOpen
  \bibfield  {author} {\bibinfo {author} {\bibfnamefont {M.}~\bibnamefont {Suzuki}},\ }\href@noop {} {\bibfield  {journal} {\bibinfo  {journal} {Journal of Mathematical Physics}\ }\textbf {\bibinfo {volume} {32}},\ \bibinfo {pages} {400} (\bibinfo {year} {1991})}\BibitemShut {NoStop}%
\bibitem [{\citenamefont {Campbell}(2019)}]{campbell2019random}%
  \BibitemOpen
  \bibfield  {author} {\bibinfo {author} {\bibfnamefont {E.}~\bibnamefont {Campbell}},\ }\href {https://journals.aps.org/prl/abstract/10.1103/PhysRevLett.123.070503} {\bibfield  {journal} {\bibinfo  {journal} {Physical review letters}\ }\textbf {\bibinfo {volume} {123}},\ \bibinfo {pages} {070503} (\bibinfo {year} {2019})}\BibitemShut {NoStop}%
\bibitem [{\citenamefont {Cerezo}\ \emph {et~al.}(2021)\citenamefont {Cerezo}, \citenamefont {Arrasmith}, \citenamefont {Babbush}, \citenamefont {Benjamin}, \citenamefont {Endo}, \citenamefont {Fujii}, \citenamefont {McClean}, \citenamefont {Mitarai}, \citenamefont {Yuan}, \citenamefont {Cincio} \emph {et~al.}}]{cerezo2021variational}%
  \BibitemOpen
  \bibfield  {author} {\bibinfo {author} {\bibfnamefont {M.}~\bibnamefont {Cerezo}}, \bibinfo {author} {\bibfnamefont {A.}~\bibnamefont {Arrasmith}}, \bibinfo {author} {\bibfnamefont {R.}~\bibnamefont {Babbush}}, \bibinfo {author} {\bibfnamefont {S.~C.}\ \bibnamefont {Benjamin}}, \bibinfo {author} {\bibfnamefont {S.}~\bibnamefont {Endo}}, \bibinfo {author} {\bibfnamefont {K.}~\bibnamefont {Fujii}}, \bibinfo {author} {\bibfnamefont {J.~R.}\ \bibnamefont {McClean}}, \bibinfo {author} {\bibfnamefont {K.}~\bibnamefont {Mitarai}}, \bibinfo {author} {\bibfnamefont {X.}~\bibnamefont {Yuan}}, \bibinfo {author} {\bibfnamefont {L.}~\bibnamefont {Cincio}},  \emph {et~al.},\ }\href {https://www.nature.com/articles/s42254-021-00348-9} {\bibfield  {journal} {\bibinfo  {journal} {Nature Reviews Physics}\ }\textbf {\bibinfo {volume} {3}},\ \bibinfo {pages} {625} (\bibinfo {year} {2021})}\BibitemShut {NoStop}%
\bibitem [{\citenamefont {Mitarai}\ \emph {et~al.}(2018)\citenamefont {Mitarai}, \citenamefont {Negoro}, \citenamefont {Kitagawa},\ and\ \citenamefont {Fujii}}]{mitarai2018quantum}%
  \BibitemOpen
  \bibfield  {author} {\bibinfo {author} {\bibfnamefont {K.}~\bibnamefont {Mitarai}}, \bibinfo {author} {\bibfnamefont {M.}~\bibnamefont {Negoro}}, \bibinfo {author} {\bibfnamefont {M.}~\bibnamefont {Kitagawa}}, \ and\ \bibinfo {author} {\bibfnamefont {K.}~\bibnamefont {Fujii}},\ }\href {https://link.aps.org/pdf/10.1103/PhysRevA.98.032309} {\bibfield  {journal} {\bibinfo  {journal} {Physical Review A}\ }\textbf {\bibinfo {volume} {98}},\ \bibinfo {pages} {032309} (\bibinfo {year} {2018})}\BibitemShut {NoStop}%
\bibitem [{\citenamefont {Benedetti}\ \emph {et~al.}(2019)\citenamefont {Benedetti}, \citenamefont {Lloyd}, \citenamefont {Sack},\ and\ \citenamefont {Fiorentini}}]{benedetti2019parameterized}%
  \BibitemOpen
  \bibfield  {author} {\bibinfo {author} {\bibfnamefont {M.}~\bibnamefont {Benedetti}}, \bibinfo {author} {\bibfnamefont {E.}~\bibnamefont {Lloyd}}, \bibinfo {author} {\bibfnamefont {S.}~\bibnamefont {Sack}}, \ and\ \bibinfo {author} {\bibfnamefont {M.}~\bibnamefont {Fiorentini}},\ }\href {https://iopscience.iop.org/article/10.1088/2058-9565/ab4eb5/pdf} {\bibfield  {journal} {\bibinfo  {journal} {Quantum Science and Technology}\ }\textbf {\bibinfo {volume} {4}},\ \bibinfo {pages} {043001} (\bibinfo {year} {2019})}\BibitemShut {NoStop}%
\bibitem [{\citenamefont {Peruzzo}\ \emph {et~al.}(2014)\citenamefont {Peruzzo}, \citenamefont {McClean}, \citenamefont {Shadbolt}, \citenamefont {Yung}, \citenamefont {Zhou}, \citenamefont {Love}, \citenamefont {Aspuru-Guzik},\ and\ \citenamefont {O’brien}}]{peruzzo2014variational}%
  \BibitemOpen
  \bibfield  {author} {\bibinfo {author} {\bibfnamefont {A.}~\bibnamefont {Peruzzo}}, \bibinfo {author} {\bibfnamefont {J.}~\bibnamefont {McClean}}, \bibinfo {author} {\bibfnamefont {P.}~\bibnamefont {Shadbolt}}, \bibinfo {author} {\bibfnamefont {M.-H.}\ \bibnamefont {Yung}}, \bibinfo {author} {\bibfnamefont {X.-Q.}\ \bibnamefont {Zhou}}, \bibinfo {author} {\bibfnamefont {P.~J.}\ \bibnamefont {Love}}, \bibinfo {author} {\bibfnamefont {A.}~\bibnamefont {Aspuru-Guzik}}, \ and\ \bibinfo {author} {\bibfnamefont {J.~L.}\ \bibnamefont {O’brien}},\ }\href {https://www.nature.com/articles/ncomms5213} {\bibfield  {journal} {\bibinfo  {journal} {Nature communications}\ }\textbf {\bibinfo {volume} {5}},\ \bibinfo {pages} {4213} (\bibinfo {year} {2014})}\BibitemShut {NoStop}%
\bibitem [{\citenamefont {McClean}\ \emph {et~al.}(2017)\citenamefont {McClean}, \citenamefont {Kimchi-Schwartz}, \citenamefont {Carter},\ and\ \citenamefont {De~Jong}}]{mcclean2017hybrid}%
  \BibitemOpen
  \bibfield  {author} {\bibinfo {author} {\bibfnamefont {J.~R.}\ \bibnamefont {McClean}}, \bibinfo {author} {\bibfnamefont {M.~E.}\ \bibnamefont {Kimchi-Schwartz}}, \bibinfo {author} {\bibfnamefont {J.}~\bibnamefont {Carter}}, \ and\ \bibinfo {author} {\bibfnamefont {W.~A.}\ \bibnamefont {De~Jong}},\ }\href {https://journals.aps.org/pra/abstract/10.1103/PhysRevA.95.042308} {\bibfield  {journal} {\bibinfo  {journal} {Physical Review A}\ }\textbf {\bibinfo {volume} {95}},\ \bibinfo {pages} {042308} (\bibinfo {year} {2017})}\BibitemShut {NoStop}%
\bibitem [{\citenamefont {Nakanishi}\ \emph {et~al.}(2019)\citenamefont {Nakanishi}, \citenamefont {Mitarai},\ and\ \citenamefont {Fujii}}]{nakanishi2019subspace}%
  \BibitemOpen
  \bibfield  {author} {\bibinfo {author} {\bibfnamefont {K.~M.}\ \bibnamefont {Nakanishi}}, \bibinfo {author} {\bibfnamefont {K.}~\bibnamefont {Mitarai}}, \ and\ \bibinfo {author} {\bibfnamefont {K.}~\bibnamefont {Fujii}},\ }\href {https://journals.aps.org/prresearch/abstract/10.1103/PhysRevResearch.1.033062} {\bibfield  {journal} {\bibinfo  {journal} {Physical Review Research}\ }\textbf {\bibinfo {volume} {1}},\ \bibinfo {pages} {033062} (\bibinfo {year} {2019})}\BibitemShut {NoStop}%
\bibitem [{\citenamefont {Schuld}\ \emph {et~al.}(2020)\citenamefont {Schuld}, \citenamefont {Bocharov}, \citenamefont {Svore},\ and\ \citenamefont {Wiebe}}]{schuld2020circuit}%
  \BibitemOpen
  \bibfield  {author} {\bibinfo {author} {\bibfnamefont {M.}~\bibnamefont {Schuld}}, \bibinfo {author} {\bibfnamefont {A.}~\bibnamefont {Bocharov}}, \bibinfo {author} {\bibfnamefont {K.~M.}\ \bibnamefont {Svore}}, \ and\ \bibinfo {author} {\bibfnamefont {N.}~\bibnamefont {Wiebe}},\ }\href {https://journals.aps.org/pra/abstract/10.1103/PhysRevA.101.032308} {\bibfield  {journal} {\bibinfo  {journal} {Physical Review A}\ }\textbf {\bibinfo {volume} {101}},\ \bibinfo {pages} {032308} (\bibinfo {year} {2020})}\BibitemShut {NoStop}%
\bibitem [{\citenamefont {Farhi}\ and\ \citenamefont {Neven}(2018)}]{farhi2018classification}%
  \BibitemOpen
  \bibfield  {author} {\bibinfo {author} {\bibfnamefont {E.}~\bibnamefont {Farhi}}\ and\ \bibinfo {author} {\bibfnamefont {H.}~\bibnamefont {Neven}},\ }\href {https://arxiv.org/abs/1802.06002} {\bibfield  {journal} {\bibinfo  {journal} {arXiv preprint arXiv:1802.06002}\ } (\bibinfo {year} {2018})}\BibitemShut {NoStop}%
\bibitem [{\citenamefont {Chen}\ \emph {et~al.}(2023)\citenamefont {Chen}, \citenamefont {Zhao},\ and\ \citenamefont {Wang}}]{chen2023near}%
  \BibitemOpen
  \bibfield  {author} {\bibinfo {author} {\bibfnamefont {R.}~\bibnamefont {Chen}}, \bibinfo {author} {\bibfnamefont {B.}~\bibnamefont {Zhao}}, \ and\ \bibinfo {author} {\bibfnamefont {X.}~\bibnamefont {Wang}},\ }\href {https://journals.aps.org/prapplied/abstract/10.1103/PhysRevApplied.20.024071} {\bibfield  {journal} {\bibinfo  {journal} {Physical Review Applied}\ }\textbf {\bibinfo {volume} {20}},\ \bibinfo {pages} {024071} (\bibinfo {year} {2023})}\BibitemShut {NoStop}%
\bibitem [{\citenamefont {Chen}\ \emph {et~al.}(2021)\citenamefont {Chen}, \citenamefont {Song}, \citenamefont {Zhao},\ and\ \citenamefont {Wang}}]{chen2021variational}%
  \BibitemOpen
  \bibfield  {author} {\bibinfo {author} {\bibfnamefont {R.}~\bibnamefont {Chen}}, \bibinfo {author} {\bibfnamefont {Z.}~\bibnamefont {Song}}, \bibinfo {author} {\bibfnamefont {X.}~\bibnamefont {Zhao}}, \ and\ \bibinfo {author} {\bibfnamefont {X.}~\bibnamefont {Wang}},\ }\href {https://iopscience.iop.org/article/10.1088/2058-9565/ac38ba/pdf?casa_token=-3q84ZGtqbMAAAAA:7yQww-nfv-UFWzXYLXaGj9GvWbM4CMRmuLgovST8MsIxh93vdjcFmNDZ-L0t8QssTU8mCm1rrTCBXPc3ajOl5jDPjhU5} {\bibfield  {journal} {\bibinfo  {journal} {Quantum Science and Technology}\ }\textbf {\bibinfo {volume} {7}},\ \bibinfo {pages} {015019} (\bibinfo {year} {2021})}\BibitemShut {NoStop}%
\bibitem [{\citenamefont {Zhao}\ \emph {et~al.}(2021)\citenamefont {Zhao}, \citenamefont {Zhao}, \citenamefont {Wang}, \citenamefont {Song},\ and\ \citenamefont {Wang}}]{zhao2021practical}%
  \BibitemOpen
  \bibfield  {author} {\bibinfo {author} {\bibfnamefont {X.}~\bibnamefont {Zhao}}, \bibinfo {author} {\bibfnamefont {B.}~\bibnamefont {Zhao}}, \bibinfo {author} {\bibfnamefont {Z.}~\bibnamefont {Wang}}, \bibinfo {author} {\bibfnamefont {Z.}~\bibnamefont {Song}}, \ and\ \bibinfo {author} {\bibfnamefont {X.}~\bibnamefont {Wang}},\ }\href {https://www.nature.com/articles/s41534-021-00496-x} {\bibfield  {journal} {\bibinfo  {journal} {npj Quantum Information}\ }\textbf {\bibinfo {volume} {7}},\ \bibinfo {pages} {159} (\bibinfo {year} {2021})}\BibitemShut {NoStop}%
\bibitem [{\citenamefont {Romero}\ \emph {et~al.}(2017)\citenamefont {Romero}, \citenamefont {Olson},\ and\ \citenamefont {Aspuru-Guzik}}]{romero2017quantum}%
  \BibitemOpen
  \bibfield  {author} {\bibinfo {author} {\bibfnamefont {J.}~\bibnamefont {Romero}}, \bibinfo {author} {\bibfnamefont {J.~P.}\ \bibnamefont {Olson}}, \ and\ \bibinfo {author} {\bibfnamefont {A.}~\bibnamefont {Aspuru-Guzik}},\ }\href {https://iopscience.iop.org/article/10.1088/2058-9565/aa8072/pdf?casa_token=rU-jrWj4n5gAAAAA:r2Dbu_cU1SXLkJPGhh7XGck77Xk3I3NFcqM82_R-krwDgmD1NGB9e4ga1KdXpgZ2ZastEZSul7rUDlDeb7o6U4QYqWdL} {\bibfield  {journal} {\bibinfo  {journal} {Quantum Science and Technology}\ }\textbf {\bibinfo {volume} {2}},\ \bibinfo {pages} {045001} (\bibinfo {year} {2017})}\BibitemShut {NoStop}%
\bibitem [{\citenamefont {Cao}\ and\ \citenamefont {Wang}(2021)}]{cao2021noise}%
  \BibitemOpen
  \bibfield  {author} {\bibinfo {author} {\bibfnamefont {C.}~\bibnamefont {Cao}}\ and\ \bibinfo {author} {\bibfnamefont {X.}~\bibnamefont {Wang}},\ }\href {https://link.aps.org/pdf/10.1103/PhysRevApplied.15.054012?casa_token=AY2gKL9JvSoAAAAA:NT1xLok5L_RTjZzMpC8bb4YbzSa4oD5K_vvOoGpmG6rOB-2u0LVTwtLBw5gvu44djTSyNPrl8tvN6DQ} {\bibfield  {journal} {\bibinfo  {journal} {Physical Review Applied}\ }\textbf {\bibinfo {volume} {15}},\ \bibinfo {pages} {054012} (\bibinfo {year} {2021})}\BibitemShut {NoStop}%
\bibitem [{\citenamefont {Wang}\ \emph {et~al.}(2021{\natexlab{a}})\citenamefont {Wang}, \citenamefont {Song},\ and\ \citenamefont {Wang}}]{wang2021variational}%
  \BibitemOpen
  \bibfield  {author} {\bibinfo {author} {\bibfnamefont {X.}~\bibnamefont {Wang}}, \bibinfo {author} {\bibfnamefont {Z.}~\bibnamefont {Song}}, \ and\ \bibinfo {author} {\bibfnamefont {Y.}~\bibnamefont {Wang}},\ }\href {https://quantum-journal.org/papers/q-2021-06-29-483/?utm_source=researcher_app&utm_medium=referral&utm_campaign=RESR_MRKT_Researcher_inbound} {\bibfield  {journal} {\bibinfo  {journal} {Quantum}\ }\textbf {\bibinfo {volume} {5}},\ \bibinfo {pages} {483} (\bibinfo {year} {2021}{\natexlab{a}})}\BibitemShut {NoStop}%
\bibitem [{\citenamefont {Zeng}\ \emph {et~al.}(2021)\citenamefont {Zeng}, \citenamefont {Cao}, \citenamefont {Zhang}, \citenamefont {Xu},\ and\ \citenamefont {Zeng}}]{zeng2021variational}%
  \BibitemOpen
  \bibfield  {author} {\bibinfo {author} {\bibfnamefont {J.}~\bibnamefont {Zeng}}, \bibinfo {author} {\bibfnamefont {C.}~\bibnamefont {Cao}}, \bibinfo {author} {\bibfnamefont {C.}~\bibnamefont {Zhang}}, \bibinfo {author} {\bibfnamefont {P.}~\bibnamefont {Xu}}, \ and\ \bibinfo {author} {\bibfnamefont {B.}~\bibnamefont {Zeng}},\ }\href {https://iopscience.iop.org/article/10.1088/2058-9565/ac11a7/meta} {\bibfield  {journal} {\bibinfo  {journal} {Quantum Science and Technology}\ }\textbf {\bibinfo {volume} {6}},\ \bibinfo {pages} {045009} (\bibinfo {year} {2021})}\BibitemShut {NoStop}%
\bibitem [{\citenamefont {Wang}\ \emph {et~al.}(2021{\natexlab{b}})\citenamefont {Wang}, \citenamefont {Li},\ and\ \citenamefont {Wang}}]{wang2021Gibbs}%
  \BibitemOpen
  \bibfield  {author} {\bibinfo {author} {\bibfnamefont {Y.}~\bibnamefont {Wang}}, \bibinfo {author} {\bibfnamefont {G.}~\bibnamefont {Li}}, \ and\ \bibinfo {author} {\bibfnamefont {X.}~\bibnamefont {Wang}},\ }\href {https://journals.aps.org/prapplied/abstract/10.1103/PhysRevApplied.16.054035} {\bibfield  {journal} {\bibinfo  {journal} {Physical Review Applied}\ }\textbf {\bibinfo {volume} {16}},\ \bibinfo {pages} {054035} (\bibinfo {year} {2021}{\natexlab{b}})}\BibitemShut {NoStop}%
\bibitem [{\citenamefont {Crawford}\ \emph {et~al.}(2021)\citenamefont {Crawford}, \citenamefont {van Straaten}, \citenamefont {Wang}, \citenamefont {Parks}, \citenamefont {Campbell},\ and\ \citenamefont {Brierley}}]{crawford2021efficient}%
  \BibitemOpen
  \bibfield  {author} {\bibinfo {author} {\bibfnamefont {O.}~\bibnamefont {Crawford}}, \bibinfo {author} {\bibfnamefont {B.}~\bibnamefont {van Straaten}}, \bibinfo {author} {\bibfnamefont {D.}~\bibnamefont {Wang}}, \bibinfo {author} {\bibfnamefont {T.}~\bibnamefont {Parks}}, \bibinfo {author} {\bibfnamefont {E.}~\bibnamefont {Campbell}}, \ and\ \bibinfo {author} {\bibfnamefont {S.}~\bibnamefont {Brierley}},\ }\href {https://quantum-journal.org/papers/q-2021-01-20-385/} {\bibfield  {journal} {\bibinfo  {journal} {Quantum}\ }\textbf {\bibinfo {volume} {5}},\ \bibinfo {pages} {385} (\bibinfo {year} {2021})}\BibitemShut {NoStop}%
\bibitem [{\citenamefont {Gokhale}\ \emph {et~al.}(2019)\citenamefont {Gokhale}, \citenamefont {Angiuli}, \citenamefont {Ding}, \citenamefont {Gui}, \citenamefont {Tomesh}, \citenamefont {Suchara}, \citenamefont {Martonosi},\ and\ \citenamefont {Chong}}]{gokhale2019minimizing}%
  \BibitemOpen
  \bibfield  {author} {\bibinfo {author} {\bibfnamefont {P.}~\bibnamefont {Gokhale}}, \bibinfo {author} {\bibfnamefont {O.}~\bibnamefont {Angiuli}}, \bibinfo {author} {\bibfnamefont {Y.}~\bibnamefont {Ding}}, \bibinfo {author} {\bibfnamefont {K.}~\bibnamefont {Gui}}, \bibinfo {author} {\bibfnamefont {T.}~\bibnamefont {Tomesh}}, \bibinfo {author} {\bibfnamefont {M.}~\bibnamefont {Suchara}}, \bibinfo {author} {\bibfnamefont {M.}~\bibnamefont {Martonosi}}, \ and\ \bibinfo {author} {\bibfnamefont {F.~T.}\ \bibnamefont {Chong}},\ }\href {https://arxiv.org/abs/1907.13623} {\bibfield  {journal} {\bibinfo  {journal} {arXiv preprint arXiv:1907.13623}\ } (\bibinfo {year} {2019})}\BibitemShut {NoStop}%
\bibitem [{\citenamefont {Verteletskyi}\ \emph {et~al.}(2020)\citenamefont {Verteletskyi}, \citenamefont {Yen},\ and\ \citenamefont {Izmaylov}}]{verteletskyi2020measurement}%
  \BibitemOpen
  \bibfield  {author} {\bibinfo {author} {\bibfnamefont {V.}~\bibnamefont {Verteletskyi}}, \bibinfo {author} {\bibfnamefont {T.-C.}\ \bibnamefont {Yen}}, \ and\ \bibinfo {author} {\bibfnamefont {A.~F.}\ \bibnamefont {Izmaylov}},\ }\href {https://pubs.aip.org/aip/jcp/article-abstract/152/12/124114/954934/Measurement-optimization-in-the-variational?redirectedFrom=fulltext} {\bibfield  {journal} {\bibinfo  {journal} {The Journal of chemical physics}\ }\textbf {\bibinfo {volume} {152}},\ \bibinfo {pages} {124114} (\bibinfo {year} {2020})}\BibitemShut {NoStop}%
\bibitem [{\citenamefont {Zhang}\ \emph {et~al.}(2022)\citenamefont {Zhang}, \citenamefont {Li},\ and\ \citenamefont {Yuan}}]{zhang2022quantum}%
  \BibitemOpen
  \bibfield  {author} {\bibinfo {author} {\bibfnamefont {X.-M.}\ \bibnamefont {Zhang}}, \bibinfo {author} {\bibfnamefont {T.}~\bibnamefont {Li}}, \ and\ \bibinfo {author} {\bibfnamefont {X.}~\bibnamefont {Yuan}},\ }\href {https://journals.aps.org/prl/abstract/10.1103/PhysRevLett.129.230504} {\bibfield  {journal} {\bibinfo  {journal} {Physical Review Letters}\ }\textbf {\bibinfo {volume} {129}},\ \bibinfo {pages} {230504} (\bibinfo {year} {2022})}\BibitemShut {NoStop}%
\bibitem [{\citenamefont {Liu}\ \emph {et~al.}(2024)\citenamefont {Liu}, \citenamefont {Sierant}, \citenamefont {Stornati}, \citenamefont {Lewenstein},\ and\ \citenamefont {P{\l}odzie{\'n}}}]{liu2024quantum}%
  \BibitemOpen
  \bibfield  {author} {\bibinfo {author} {\bibfnamefont {Y.}~\bibnamefont {Liu}}, \bibinfo {author} {\bibfnamefont {P.}~\bibnamefont {Sierant}}, \bibinfo {author} {\bibfnamefont {P.}~\bibnamefont {Stornati}}, \bibinfo {author} {\bibfnamefont {M.}~\bibnamefont {Lewenstein}}, \ and\ \bibinfo {author} {\bibfnamefont {M.}~\bibnamefont {P{\l}odzie{\'n}}},\ }\href {https://arxiv.org/abs/2405.03338} {\bibfield  {journal} {\bibinfo  {journal} {arXiv preprint arXiv:2405.03338}\ } (\bibinfo {year} {2024})}\BibitemShut {NoStop}%
\bibitem [{\citenamefont {Nielsen}\ and\ \citenamefont {Chuang}(2010)}]{nielsen2010quantum}%
  \BibitemOpen
  \bibfield  {author} {\bibinfo {author} {\bibfnamefont {M.~A.}\ \bibnamefont {Nielsen}}\ and\ \bibinfo {author} {\bibfnamefont {I.~L.}\ \bibnamefont {Chuang}},\ }\href@noop {} {\emph {\bibinfo {title} {Quantum computation and quantum information}}}\ (\bibinfo  {publisher} {Cambridge university press},\ \bibinfo {year} {2010})\BibitemShut {NoStop}%
\bibitem [{\citenamefont {Gulbahar}(2021)}]{gulbahar2021k}%
  \BibitemOpen
  \bibfield  {author} {\bibinfo {author} {\bibfnamefont {B.}~\bibnamefont {Gulbahar}},\ }\href {https://arxiv.org/pdf/2111.04359} {\bibfield  {journal} {\bibinfo  {journal} {arXiv preprint arXiv:2111.04359}\ } (\bibinfo {year} {2021})}\BibitemShut {NoStop}%
\bibitem [{\citenamefont {Childs}\ \emph {et~al.}(2019)\citenamefont {Childs}, \citenamefont {Ostrander},\ and\ \citenamefont {Su}}]{childs2019faster}%
  \BibitemOpen
  \bibfield  {author} {\bibinfo {author} {\bibfnamefont {A.~M.}\ \bibnamefont {Childs}}, \bibinfo {author} {\bibfnamefont {A.}~\bibnamefont {Ostrander}}, \ and\ \bibinfo {author} {\bibfnamefont {Y.}~\bibnamefont {Su}},\ }\href {https://quantum-journal.org/papers/q-2019-09-02-182/} {\bibfield  {journal} {\bibinfo  {journal} {Quantum}\ }\textbf {\bibinfo {volume} {3}},\ \bibinfo {pages} {182} (\bibinfo {year} {2019})}\BibitemShut {NoStop}%
\bibitem [{\citenamefont {Childs}\ \emph {et~al.}(2018)\citenamefont {Childs}, \citenamefont {Maslov}, \citenamefont {Nam}, \citenamefont {Ross},\ and\ \citenamefont {Su}}]{childs2018toward}%
  \BibitemOpen
  \bibfield  {author} {\bibinfo {author} {\bibfnamefont {A.~M.}\ \bibnamefont {Childs}}, \bibinfo {author} {\bibfnamefont {D.}~\bibnamefont {Maslov}}, \bibinfo {author} {\bibfnamefont {Y.}~\bibnamefont {Nam}}, \bibinfo {author} {\bibfnamefont {N.~J.}\ \bibnamefont {Ross}}, \ and\ \bibinfo {author} {\bibfnamefont {Y.}~\bibnamefont {Su}},\ }\href {https://www.pnas.org/doi/full/10.1073/pnas.1801723115?doi=10.1073%2Fpnas.1801723115} {\bibfield  {journal} {\bibinfo  {journal} {Proceedings of the National Academy of Sciences}\ }\textbf {\bibinfo {volume} {115}},\ \bibinfo {pages} {9456} (\bibinfo {year} {2018})}\BibitemShut {NoStop}%
\bibitem [{\citenamefont {McClean}\ \emph {et~al.}(2018)\citenamefont {McClean}, \citenamefont {Boixo}, \citenamefont {Smelyanskiy}, \citenamefont {Babbush},\ and\ \citenamefont {Neven}}]{mcclean2018barren}%
  \BibitemOpen
  \bibfield  {author} {\bibinfo {author} {\bibfnamefont {J.~R.}\ \bibnamefont {McClean}}, \bibinfo {author} {\bibfnamefont {S.}~\bibnamefont {Boixo}}, \bibinfo {author} {\bibfnamefont {V.~N.}\ \bibnamefont {Smelyanskiy}}, \bibinfo {author} {\bibfnamefont {R.}~\bibnamefont {Babbush}}, \ and\ \bibinfo {author} {\bibfnamefont {H.}~\bibnamefont {Neven}},\ }\href {https://www.nature.com/articles/s41467-018-07090-4} {\bibfield  {journal} {\bibinfo  {journal} {Nature communications}\ }\textbf {\bibinfo {volume} {9}},\ \bibinfo {pages} {4812} (\bibinfo {year} {2018})}\BibitemShut {NoStop}%
\bibitem [{\citenamefont {Marrero}\ \emph {et~al.}(2021)\citenamefont {Marrero}, \citenamefont {Kieferov{\'a}},\ and\ \citenamefont {Wiebe}}]{marrero2021entanglement}%
  \BibitemOpen
  \bibfield  {author} {\bibinfo {author} {\bibfnamefont {C.~O.}\ \bibnamefont {Marrero}}, \bibinfo {author} {\bibfnamefont {M.}~\bibnamefont {Kieferov{\'a}}}, \ and\ \bibinfo {author} {\bibfnamefont {N.}~\bibnamefont {Wiebe}},\ }\href {https://journals.aps.org/prxquantum/abstract/10.1103/PRXQuantum.2.040316} {\bibfield  {journal} {\bibinfo  {journal} {PRX Quantum}\ }\textbf {\bibinfo {volume} {2}},\ \bibinfo {pages} {040316} (\bibinfo {year} {2021})}\BibitemShut {NoStop}%
\bibitem [{\citenamefont {kai Zhang}\ \emph {et~al.}(2023)\citenamefont {kai Zhang}, \citenamefont {Zhu}, \citenamefont {Jing},\ and\ \citenamefont {Wang}}]{zhang2023statistical}%
  \BibitemOpen
  \bibfield  {author} {\bibinfo {author} {\bibfnamefont {H.}~\bibnamefont {kai Zhang}}, \bibinfo {author} {\bibfnamefont {C.}~\bibnamefont {Zhu}}, \bibinfo {author} {\bibfnamefont {M.}~\bibnamefont {Jing}}, \ and\ \bibinfo {author} {\bibfnamefont {X.}~\bibnamefont {Wang}},\ }\href {https://arxiv.org/abs/2309.14980} {\enquote {\bibinfo {title} {Statistical analysis of quantum state learning process in quantum neural networks},}\ } (\bibinfo {year} {2023}),\ \Eprint {http://arxiv.org/abs/2309.14980} {arXiv:2309.14980 [quant-ph]} \BibitemShut {NoStop}%
\bibitem [{\citenamefont {Bravyi}\ \emph {et~al.}(2017)\citenamefont {Bravyi}, \citenamefont {Gambetta}, \citenamefont {Mezzacapo},\ and\ \citenamefont {Temme}}]{bravyi2017tapering}%
  \BibitemOpen
  \bibfield  {author} {\bibinfo {author} {\bibfnamefont {S.}~\bibnamefont {Bravyi}}, \bibinfo {author} {\bibfnamefont {J.~M.}\ \bibnamefont {Gambetta}}, \bibinfo {author} {\bibfnamefont {A.}~\bibnamefont {Mezzacapo}}, \ and\ \bibinfo {author} {\bibfnamefont {K.}~\bibnamefont {Temme}},\ }\href {https://arxiv.org/abs/1701.08213} {\bibfield  {journal} {\bibinfo  {journal} {arXiv preprint arXiv:1701.08213}\ } (\bibinfo {year} {2017})}\BibitemShut {NoStop}%
\bibitem [{\citenamefont {Setia}\ \emph {et~al.}(2020)\citenamefont {Setia}, \citenamefont {Chen}, \citenamefont {Rice}, \citenamefont {Mezzacapo}, \citenamefont {Pistoia},\ and\ \citenamefont {Whitfield}}]{setia2020reducing}%
  \BibitemOpen
  \bibfield  {author} {\bibinfo {author} {\bibfnamefont {K.}~\bibnamefont {Setia}}, \bibinfo {author} {\bibfnamefont {R.}~\bibnamefont {Chen}}, \bibinfo {author} {\bibfnamefont {J.~E.}\ \bibnamefont {Rice}}, \bibinfo {author} {\bibfnamefont {A.}~\bibnamefont {Mezzacapo}}, \bibinfo {author} {\bibfnamefont {M.}~\bibnamefont {Pistoia}}, \ and\ \bibinfo {author} {\bibfnamefont {J.~D.}\ \bibnamefont {Whitfield}},\ }\href {https://pubs.acs.org/doi/pdf/10.1021/acs.jctc.0c00113?casa_token=Pzzi17VtHDwAAAAA:s6iL94jif3NgOCJA6bGhSfJt_iqh7TMHUZEd9GX71h2RsbhG55Vx9Bl3YHhQOlzPvMV_s1FvCpmwJlJD} {\bibfield  {journal} {\bibinfo  {journal} {Journal of Chemical Theory and Computation}\ }\textbf {\bibinfo {volume} {16}},\ \bibinfo {pages} {6091} (\bibinfo {year} {2020})}\BibitemShut {NoStop}%
\bibitem [{\citenamefont {Aguado}\ \emph {et~al.}(2021)\citenamefont {Aguado}, \citenamefont {Roncero},\ and\ \citenamefont {Sanz-Sanz}}]{aguado2021three}%
  \BibitemOpen
  \bibfield  {author} {\bibinfo {author} {\bibfnamefont {A.}~\bibnamefont {Aguado}}, \bibinfo {author} {\bibfnamefont {O.}~\bibnamefont {Roncero}}, \ and\ \bibinfo {author} {\bibfnamefont {C.}~\bibnamefont {Sanz-Sanz}},\ }\href {https://pubs.rsc.org/en/content/articlelanding/2021/cp/d0cp04100a} {\bibfield  {journal} {\bibinfo  {journal} {Physical Chemistry Chemical Physics}\ }\textbf {\bibinfo {volume} {23}},\ \bibinfo {pages} {7735} (\bibinfo {year} {2021})}\BibitemShut {NoStop}%
\bibitem [{\citenamefont {Yen}\ \emph {et~al.}(2020)\citenamefont {Yen}, \citenamefont {Verteletskyi},\ and\ \citenamefont {Izmaylov}}]{yen2020measuring}%
  \BibitemOpen
  \bibfield  {author} {\bibinfo {author} {\bibfnamefont {T.-C.}\ \bibnamefont {Yen}}, \bibinfo {author} {\bibfnamefont {V.}~\bibnamefont {Verteletskyi}}, \ and\ \bibinfo {author} {\bibfnamefont {A.~F.}\ \bibnamefont {Izmaylov}},\ }\href {https://pubs.acs.org/doi/pdf/10.1021/acs.jctc.0c00008?casa_token=48eCVafAMI4AAAAA:wi-slvynIm-4SBmNtpwVCXpSADcSig5mgNpawB1fCe-C7epNOL4b5AHHzns-sG03L-CwLoi0BryHgyYC} {\bibfield  {journal} {\bibinfo  {journal} {Journal of chemical theory and computation}\ }\textbf {\bibinfo {volume} {16}},\ \bibinfo {pages} {2400} (\bibinfo {year} {2020})}\BibitemShut {NoStop}%
\bibitem [{\citenamefont {Romero}\ \emph {et~al.}(2018)\citenamefont {Romero}, \citenamefont {Babbush}, \citenamefont {McClean}, \citenamefont {Hempel}, \citenamefont {Love},\ and\ \citenamefont {Aspuru-Guzik}}]{romero2018strategies}%
  \BibitemOpen
  \bibfield  {author} {\bibinfo {author} {\bibfnamefont {J.}~\bibnamefont {Romero}}, \bibinfo {author} {\bibfnamefont {R.}~\bibnamefont {Babbush}}, \bibinfo {author} {\bibfnamefont {J.~R.}\ \bibnamefont {McClean}}, \bibinfo {author} {\bibfnamefont {C.}~\bibnamefont {Hempel}}, \bibinfo {author} {\bibfnamefont {P.~J.}\ \bibnamefont {Love}}, \ and\ \bibinfo {author} {\bibfnamefont {A.}~\bibnamefont {Aspuru-Guzik}},\ }\href {https://iopscience.iop.org/article/10.1088/2058-9565/aad3e4} {\bibfield  {journal} {\bibinfo  {journal} {Quantum Science and Technology}\ }\textbf {\bibinfo {volume} {4}},\ \bibinfo {pages} {014008} (\bibinfo {year} {2018})}\BibitemShut {NoStop}%
\end{thebibliography}%
